\documentclass[12pt,english]{article}
\usepackage[T1]{fontenc}
\usepackage[latin9]{inputenc}
\usepackage{xcolor}
\usepackage{babel}
\usepackage{float}
\usepackage{enumitem}
\usepackage{amsmath}
\usepackage{amsthm}
\usepackage{amssymb}
\usepackage{graphicx}
\usepackage[a4paper]{geometry}
\usepackage{setspace}
\usepackage[authoryear]{natbib}
\usepackage{microtype}
\usepackage[pdfusetitle,
 bookmarks=true,bookmarksnumbered=false,bookmarksopen=false,
 breaklinks=false,pdfborder={0 0 0},pdfborderstyle={},backref=false,colorlinks=false]
 {hyperref}

\makeatletter

\providecommand{\tabularnewline}{\\}
\floatstyle{ruled}
\newfloat{algorithm}{tbp}{loa}
\providecommand{\algorithmname}{Algorithm}
\floatname{algorithm}{\protect\algorithmname}

\theoremstyle{plain}
\newtheorem{lem}{\protect\lemmaname}
\theoremstyle{plain}
\newtheorem{assumption}{\protect\assumptionname}
\theoremstyle{plain}
\newtheorem{thm}{\protect\theoremname}

\usepackage{lmodern}
\usepackage{algorithm}
\usepackage{algpseudocode}
\usepackage{mathtools}

\ifdefined\showcaptionsetup
 \PassOptionsToPackage{caption=false}{subfig}
\fi
\usepackage{subfig}
\makeatother

\providecommand{\assumptionname}{Assumption}
\providecommand{\lemmaname}{Lemma}
\providecommand{\theoremname}{Theorem}

\begin{document}
\title{Generative Predictive Distributions for Time Series\thanks{The authors thank Ministerio de Ciencia, Innovación y Universidades
through grant CEX2021-001181-M (Llorens-Terrazas) and the Research
Council of Finland (Meitz) for financial support, and CSC -- IT Center
for Science, Finland, for computational resources during early stages
of this research. We would also like to thank Caio Almeida, Dante
Amengual, Giovanni Ballarin, Christian Brownlees, Valentina Corradi,
Juan Carlos Escanciano, Gustavo Freire, Lyudmila Grigoryeva, Stefán
Guðmundsson, Daniel Gutknecht, Adam Lee, Alessandra Luati, Eduardo
Fonseca Mendes, Rasmus Søndergaard Pedersen, Anders Rahbek, Alessio
Sancetta, André B.M. Souza, Jesper Riis-Vestergaard Sørensen, Emilio
Zanetti Chini, and Stefan Voigt for useful discussions and for providing
numerous helpful comments, and Christian Dahl and Emil Sørensen for
sharing their code with us. Contact addresses: Jordi Llorens-Terrazas,
Department of Economics, Universidad Carlos III de Madrid, 28903 Getafe
(Madrid); e-mail: jordi.llorens@uc3m.es. Mika Meitz, Discipline of
Economics, University of Helsinki, P. O. Box 17, FI--00014 University
of Helsinki, Finland; e-mail: mika.meitz@helsinki.fi.}}
\author{Jordi Llorens-Terrazas\\
Universidad Carlos III de Madrid \and Mika Meitz\\
University of Helsinki}
\date{\vspace{10bp}
June 15, 2026}
\maketitle
\begin{abstract}
\begin{singlespace}
\noindent We propose a flexible framework for modeling the predictive
distributions of nonlinear, possibly multivariate time series. Our
approach expresses a general predictive distribution in an appropriate
generative representation that is based on a folklore result from
measure theoretic probability. This representation provides a direct
simulation-based approximation to the predictive distribution, enabling
straightforward computation of forecasts for the conditional mean
and variance, fan charts, value at risk, expected shortfall, joint
tail risks, and other quantities of interest. We estimate this generative
representation using a version of conditional generative adversarial
networks and provide a formal statistical analysis of estimation under
weak temporal dependence. Specifically, estimation is expressed as
a particular minimax problem and we establish consistency of its approximate
solutions in Hausdorff distance. The empirical relevance of the approach
is illustrated using applications to equity returns, realized variance,
and realized covariances. The proposed method is also computationally
manageable, with estimation in our applications taking approximately
one minute on a standard laptop.
\end{singlespace}
\end{abstract}
\bigskip{}

\begin{singlespace}
\end{singlespace}

\bigskip{}

\begin{singlespace}
 \vfill{}

\end{singlespace}

\pagebreak{}

\section{Introduction}

The predictive distribution of a (possibly multivariate) time series
provides a complete description of the uncertainty surrounding its
future value(s), and is a key input for financial institutions, businesses,
and governments \citep{Corradi:Swanson:2006}. Classical applications
include fan charts, Value-at-Risk analysis \citep{Duffie:Pan:1997},
portfolio management \citep{Guidolin:Timmermann:2008}, and systemic
risk measurement \citep{Brownlees:Engle:2017}. In many settings,
the end goal is not merely to produce forecasts but to make decisions
that involve future values of a time series, which are unknown to
the practitioner at the time decisions are made. A case in point is
a risk-averse investor who must decide what proportion of wealth to
allocate to a risky asset, see e.g. \citet{Kandel:Stambaugh:1996}.
In such cases, the researcher needs a method for sampling from the
predictive distribution. The typical solution is to postulate a model
describing the evolution of the time series and combining this with
a convenient distributional assumption, such as Gaussianity. Common
approaches used for specifying the predictive distribution include
observation-driven models (such as autoregressive moving average (ARMA)
or generalized autoregressive conditionally heteroskedastic (GARCH)
models), parameter-driven models (such as state-space or stochastic
volatility models), semi- or non-parametric approaches, and various
nonlinear models (such as mixture or Markov switching models). These
approaches will be briefly reviewed in Section \ref{sec:literature}. 

In this paper, we consider a different approach to estimating the
predictive distribution of a time series. This approach is inspired
by a recent, novel method for sampling from a general conditional
distribution that was recently proposed by \citet{Zhou:etal:2023}
and further discussed by \citet{song:etal:202X}. In a nutshell, these
authors introduced a method for sampling from a conditional distribution
by first expressing it in a certain generative representation, and
then estimating this representation using a suitable formulation of
so-called generative adversarial networks (GANs). GAN is a machine
learning method introduced by \citet{Goodfellow:etal:2014}; we will
review it in Section \ref{sec:literature}. This generative representation
of the predictive distribution is an explicit function or formula
that takes two main arguments: (i) a vector of conditioning variables,
typically including lags of the time series, and (ii) an innovation
drawn from a convenient distribution for sampling, such as a standard
Gaussian. With such a representation in hand, sampling from the predictive
distribution of any nonlinear transformation of the time series is
straightforward. In other words, given conditioning variables, this
representation specifies in closed form a map (or transport) from
a Gaussian distribution to the predictive distribution of interest.
Once estimated, the generative representation provides a direct simulation-based
approximation to the predictive distribution, enabling straightforward
computation of nonlinear transformations, risk measures, and/or optimal
decision rules.

The aims and contributions of this paper are fourfold. First, we extend
the previous work of \citet{Zhou:etal:2023} and \citet{song:etal:202X}
to the time series setting: while their work concerns conditional
distributions of a variable $Y_{t}$ given another variable $X_{t}$
where the pairs $(Y_{t},X_{t})$ are independent and identically distributed
(IID) over time, our focus is on economic and financial time series
data that are not IID. We express the predictive distribution of a
time series in a generative representation and in doing so develop
a general framework for predictive distribution modeling in time series,
accommodating both uni- and multivariate outcomes, either one-step
or multiple-steps ahead predictive distributions, nonlinear dynamics,
non-Gaussian features, and potential additional covariates. Linear
Gaussian autoregressive models arise as special cases, while richer
specifications are obtained by approximating the generative representation
via a complicated nonlinear function such as a neural network. For
ease of discussion, we call the proposed approach the Generative Predictive
Distribution (GPD) method.

Second, we provide a formal statistical analysis of estimation for
the GPD methodology. The generative representation of the predictive
distribution is estimated using GANs (details will be given in Section
\ref{sec:gpd}). As in the GAN literature, we will call the generative
representation simply the generator for short. This generator is estimated
using an adversarial criterion that compares observed data to simulated
data via a flexible discriminator. This estimation problem can be
expressed as a particular minimax problem, with such problems typically
admitting a multitude of solutions and being solved only approximately.
In practice, both the generator and the discriminator are complicated
neural networks that are inherently non-identifiable, leading to a
large number of solutions. In our theoretical results, we treat both
population and sample solutions as set-valued objects and establish
Hausdorff consistency of approximate estimators under weak temporal
dependence. These developments build on previous results of \citet{Meitz:2024}
who considered related questions in minimax estimation problems with
IID data. Although statistical properties of estimators such as consistency
are questions that are typically not of interest in the machine learning
literature, for econometricians and statisticians our consistency
result is reassuring in that estimation will be consistent even under
the empirically relevant case of multiple solutions. In particular,
our theoretical framework covers the setups considered in our empirical
applications. Our proof relies on empirical process methods developed
by \citet{Arcones:Yu:1994}, which we verify under more primitive
conditions and can simultaneously accommodate temporal dependence,
unbounded data, and non-differentiability of the class of functions.
The latter is important so as to allow for activation functions that
are not everywhere differentiable, such as rectified linear units
(ReLU), that are commonly employed in the machine learning literature.

Third, we aim to strike the right balance between making our methodology
accessible to readers new to machine learning, computational ease
of the proposed methods, and empirical relevance. For clarity and
transparency of presentation, we focus on the commonly used multilayer
perceptron (MLP) networks with ReLU activation \citep{Lecun:Bengio:Hinton:2015}
for the generator and discriminator to fix ideas, though our results
apply to a wider generality of models and are not restricted to neural
networks. We find that even if care is needed in specifying the networks,
there is also no need to go to great lengths in order to obtain reasonable
practical results. Following our GPD recipe, estimation in our applications
can be carried out on a standard laptop in approximately one minute,
which stands in stark contrast with the computation times in the majority
of existing applications based on GANs.

Fourth, we demonstrate the empirical relevance and performance of
our GPD methodology in three different substantive applications in
financial econometrics: (i) In an application to portfolio allocation
based on asset returns, the GPD-based predictive distribution captures
state-dependent asymmetries and tail behavior in S\&P 500 monthly
returns that lead to nontrivial variation in optimal allocations for
an investor with constant relative risk aversion utility; (ii) in
an application to out-of-sample S\&P 500 realized variance forecasting,
GPD performs competitively relative to standard benchmarks, with particularly
strong results when forecasting variance in levels, where departures
from Gaussianity are most pronounced; (iii) in an application to modeling
the realized covariance between JP Morgan and Bank of America, GPD
produces predictive distributions that reproduce both marginal features
and dependence patterns which are consistent with a non-Gaussian distribution.
We emphasize that across these three different applications, we employ
the same algorithm, criterion function, and neural network architectures.
This illustrates the ability of our approach to handle a wide variety
of tasks.

The rest of the paper is organized as follows. In Section \ref{sec:literature},
we discuss the relation of our paper to existing literature. Section
\ref{sec:gpd} presents the proposed GPD methodology in detail. Section
\ref{sec:Estimation} considers practical estimation as well as consistency
in the Hausdorff sense. Section \ref{sec:empirical} develops three
empirical applications: dynamic asset allocation based on stock returns,
realized variance forecasting, and realized covariance modeling. Section
\ref{sec:conclusion} concludes, while the Appendix contains all proofs
and some technical details.

\section{Relation to existing literature\label{sec:literature}}

To place our paper in context with existing literature we next discuss
its relation to a few different strands of literature. Readers primarily
interested in the contributions of the present paper can proceed directly
to Section 3. 

\subsection{The \citet{Zhou:etal:2023} and \citet{song:etal:202X} papers}

As mentioned in the Introduction, our paper builds on the previous
works of \citet{Zhou:etal:2023} and \citet{song:etal:202X}. Compared
to these works, our paper has a number of differences and additional
contributions. First, their work is in an IID setting and concerns
the conditional distributions of $Y_{t}$ given $X_{t}$ with the
pairs $(Y_{t},X_{t})$ being IID over time. In our Section \ref{subsec:generator}
we show how their idea of a generative representation of a conditional
distribution can be extended to a time series setting (see our Lemma
\ref{lem:representation}). Second, while our formulations of the
generator and the discriminator are quite similar to these two earlier
papers, we consider a different minimax-type criterion function that
estimation is based on. The technical details will be given in Section
\ref{subsec:criterion}, but we already note that our formulation
relies on a more recent relativistic formulation of GANs that has
lately been promoted as more stable and well-behaving from an optimization
point of view. This alternative formulation proves useful also in
our theoretical derivations. As a third additional contribution compared
to these earlier works, we study the statistical properties of estimators
in our GPD framework and establish Hausdorff consistency. Fourth,
while the applications considered by \citet{Zhou:etal:2023} and \citet{song:etal:202X}
are more suited for a machine learning audience, we consider three
different substantive applications that are of direct interest to
econometricians.

\subsection{Generative adversarial networks (GANs)}

As will be discussed in Section \ref{sec:gpd}, the formulation and
estimation of our GPD method is based on a version of GANs. To support
intuition and facilitate later discussions, we first informally describe
what GANs are. A GAN is a machine learning method introduced in \citet{Goodfellow:etal:2014}
and its basic purpose is to learn how to generate synthetic data based
on a training set of real-world examples (of, say, images). The original
GAN problem can be expressed as the minimax problem 
\[
\adjustlimits\inf_{\gamma\in\Gamma}\sup_{\delta\in\Delta}f_{\text{Goodfellow}}(\gamma,\delta)\quad\text{with}\quad f_{\text{Goodfellow}}(\gamma,\delta)=\mathbb{E}[\ln(\text{Dis}(Y,\delta))]+\mathbb{E}[\ln(1-\text{Dis}(\text{Gen}(Z,\gamma),\delta))]
\]
(this exact formulation will not be used in our paper but presenting
it may aid the readers' intuition). Here the random vector $Y$ represents
underlying real data whose distribution remains unknown to us. The
generator $\text{Gen}(\cdot,\gamma)$ is typically a complicated neural
network that depends on parameters $\gamma\in\Gamma$ that are tuned
with the aim that the distribution of the output synthetic data $\text{Gen}(Z,\gamma)$
closely mimics that of the real data $Y$ (with $Z$ denoting a random
vector of input noise variables). To assess the quality of the synthetic
data produced, a discriminating mechanism $\text{Dis}(\cdot,\delta)$
(a complicated neural network with parameters $\delta\in\Delta$ to
be estimated) indicates how likely it is that an input, whether original
data $Y$ or a replica $\text{Gen}(Z,\gamma)$, is real data. The
formulation we use in this paper is similar to this original formulation
but with several important differences to be discussed later in Section
\ref{sec:gpd}. 

A central conceptual distinction in the GAN literature is whether
the generative process is unconditional or conditional. In our approach,
the focus is on generating the conditional predictive distribution
of future observables given past information. Suppose our observed
data consists of a response variable $Y$ and covariates $X$. An
unconditional GAN aims to learn the joint distribution $P_{Y,X}$.
In this setting, the covariates are not fixed; rather, the generator
takes only the noise vector $Z$ as input and samples the full data
pair, namely $(Y_{\text{fake}},X_{\text{fake}})=\text{Gen}(Z,\gamma)$.
In contrast, a conditional GAN \citep{Mirza:Osindero:2014} targets
the conditional distribution $P_{Y\mid X}$. The covariates $X$ are
treated as given information and are fed into both networks. The generator
function becomes $\text{Gen}(Z,X,\gamma)$, producing only a synthetic
response $Y_{\text{fake}}$ conditioned on the observed $X$. Similarly,
the discriminator evaluates the pair via $\text{Dis}(Y,X,\delta)$,
judging not just whether $Y$ looks realistic in isolation, but whether
$Y$ is a plausible realization given the specific covariates $X$.
This conditional version of GANs is the one used in this paper.

We are not the first to apply GANs in time series nor in economics,
nor to study their statistical properties. In the economics literature,
\citet{Kaji:Manresa:Pouliot:2023}\textcolor{red}{{} }illustrated how
the GAN principle serves as a general framework for estimation and
inference for structural models. Other clever use cases in economics
and finance include design of empirically relevant Monte Carlo simulations
\citep{Athey:Imbens:Metzger:Munro:2024}, time series bootstrapping
\citep{Dahl:Sorensen:2022}, and asset pricing \citep{Chen:Pelger:Zhu:2024}.
Statistical aspects of GANs have been studied for instance in \citet{Biau:etal:2020},
\citet{Meitz:2024}, and \citet{puchkin2024rates}. Various theoretical
aspects of GANs have been studied in the machine learning literature;
these papers are too numerous to survey here but the three papers
just cited contain many useful references. 

When transitioning from IID data to time series, the GAN framework
requires careful adaptation. It is reasonable to expect unconditional
GAN algorithms to still learn the unconditional distribution of a
stationary time series, but substantial architectural modifications
are needed if the goal is to generate synthetic series that preserve
also the dependence structure of the true data generating process.
Prominent examples of such time-series specific algorithms include
TimeGAN \citep{Yoon:Jarrett:VanderSchaar:2019} and QuantGAN \citep{Wiese:etal:2020}.
In contrast to these works, our approach is more closely aligned with
\citet{Zhou:etal:2023} and \citet{song:etal:202X}. Because our method
focuses on sampling directly from the predictive distribution ---
effectively acting as a conditional GAN where $X$ represents time
lags and/or other covariates --- virtually no additional effort is
needed compared to the standard GAN algorithm implementation. The
main difference between these two papers and ours lies in the specific
loss functions we employ, as well as allowing for weak temporal dependence.
We also mention \citet{Haas:Richter:2020}, who provide statistical
theory for conditional and unconditional Wasserstein GAN for dependent
data. A difference to their approach is that the criterion function
in our method does not depend on tuning parameters.

\subsection{Existing methods for predictive distributions of time series}

A large part of the time-series literature has modeled the predictive
distribution through observation-driven models (in the terminology
of \citealp{Cox:1981}), where conditional distribution features such
as moments are updated deterministically as a function of past data.
The classical location-scale paradigm, exemplified by ARMA models
combined with (G)ARCH models \citep{Engle:1982,Bollerslev:1986},
has been extended beyond the conditional mean and variance to allow
for time-varying skewness, kurtosis, and other shape parameters \citep{Hansen:1994}.
In a finance context, there is considerable evidence that the unconditional
distribution of returns cannot be characterized by mean and variance
alone. \citet{Harvey:Siddique:2000} argue that, everything else equal,
investors should prefer right-skewed portfolios to left-skewed ones.
This has motivated either the adoption of richer innovation distributions
or increasingly sophisticated conditional specifications, with a shift
towards modeling the full predictive distribution. A large number
of these contributions can also be regarded as special cases of the
Generalized Autoregressive Score (GAS) framework put forward by \citet{Creal:Koopman:Lucas:2013}.
In the multivariate setting, convenient formulations for the predictive
distribution can also be achieved using conditional copulas (\citealp{patton2013copula}).

In parallel, a substantial body of work adopts a parameter-driven
perspective, in which distributional dynamics arise from latent states.
State-space models, stochastic volatility, and related specifications
allow for a high degree of flexibility by introducing unobserved components
that govern the evolution of the conditional distribution \citep{Harvey:1990}.
While this framework is attractive from a modeling standpoint, inference
and prediction typically require filtering and/or smoothing techniques.
Exact results are available only in special cases, such as linear-Gaussian
systems handled by the Kalman filter; for general nonlinear, non-Gaussian
processes, practitioners concerned about the predictive distribution
rely on approximate or simulation-based methods, including nonlinear
filtering and particle-based algorithms \citep{Doucet:DeFreitas:Gordon:2001}.

Another strand of the literature seeks to avoid strict parametric
assumptions by modeling conditional distributions in a semiparametric
or nonparametric fashion. Kernel-based conditional density estimation,
as developed for dependent data by \citet[chapter 6]{Li:Racine:2007},
provides flexible distributional estimates at the cost of potential
curse of dimensionality issues. Another popular approach is to rely
on mixture-based models, where the conditional density is represented
as a combination of simpler components. Markov-switching models \citep{Hamilton:1989}
and more recent mixture-based dynamic models (e.g., \citealp{Wong:Li:2000,Kalliovirta:Meitz:Saikkonen:2016})
offer a compromise between flexibility and structure.

More recently, the literature has increasingly turned to modern machine
learning methods to tackle conditional density estimation (CDE) while
mitigating the traditional curse of dimensionality. For example, \citet{Izbicki:Lee:2016}
introduce a flexible, nonparametric methodology for CDE using orthogonal
basis expansions that converts the density estimation problem into
a more tractable regression task. Alongside such advancements, deep
learning techniques --- such as mixture density networks, normalizing
flows, and conditional variational autoencoders (see, e.g., \citealp{Rothfuss:2019})
--- have become prominent tools for modeling complex conditional
distributions without restrictive parametric forms. While these contemporary
methods share our objective of flexibly capturing data dependencies,
our adversarial framework distinguishes itself by circumventing explicit
likelihood evaluation or density estimation entirely, focusing instead
on directly generating samples from the predictive distribution.

Many of these observation-driven and parameter-driven frameworks admit
Bayesian counterparts, in which uncertainty about parameters and latent
states is propagated into the predictive distribution. One advantage
of the Bayesian framework is that it offers a logically coherent framework
to incorporate parameter uncertainty \citep{Geweke:Amisano:2010,Geweke:Amisano:2011}.
Except in conjugate settings, Bayesian inference relies on MCMC or
related simulation techniques, with predictive analysis proceeding
via posterior sampling. This simulation-based view emphasizes that
once samples from the relevant distribution are available, any byproduct
of the posterior distribution, including the predictive distribution,
can be obtained without substantial additional effort. While sharing
this emphasis on sampling-based inference for predictive distributions,
we rely on adversarial estimation rather than posterior-based simulation,
yielding models in which the predictive distribution can be sampled
directly once estimation is complete, without recourse to filtering
or posterior simulation.

\subsection{Simulation-based estimation}

We also briefly point out that the GPD framework can be interpreted
as a simulation-based estimation method that in a way generalizes
classical approaches such as the simulated method of moments (SMM;
\citealp{Pakes:Pollard:1989}), indirect inference \citep{Gourieroux:Monfort:Renault:1993},
and the efficient method of moments \citep{Gallant:Tauchen:1996}.
In these traditional frameworks, the researcher aims to estimate the
parameters $\gamma$ of a generative structural model by matching
features of the observed data with those of simulated data.

To make this connection concrete, consider the standard SMM approach.
The estimator minimizes a distance metric between a set of user-specified,
fixed moment functions $m(\cdot)$ evaluated on real and simulated
data:
\[
\min_{\gamma}\,\|\mathbb{E}\left[m(Y,X)\right]-\mathbb{E}\left[m(\text{Gen}(Z,X,\gamma),X)\right]\|_{W}^{2},
\]
where $\|u\|_{W}^{2}=u'Wu$ and $W$ is a positive definite weighting
matrix. The limitation of SMM and related methods is the ex-ante selection
of the moment functions $m(\cdot)$ --- or the choice of an auxiliary
model in indirect inference. A poorly specified auxiliary model or
an arbitrary selection of moments may fail to capture the most informative
features of the true data distribution.

As pointed out by \citet{Kaji:Manresa:Pouliot:2023}, adversarial
estimation overcomes this limitation by nesting the matching problem
within a minimax game. Instead of relying on fixed moments, the discriminator
acts as an adaptive test function that actively searches for the features
that most effectively distinguish between the true data and the simulated
outcomes. The estimation problem transforms into finding a generator
that performs well against a worst-case discriminator chosen from
a broad class $\mathcal{D}$:
\[
\min_{\gamma}\sup_{D\in\mathcal{D}}\mathbb{E}[\ln D(Y,X)]+\mathbb{E}[\ln(1-D(\text{Gen}(Z,X,\gamma),X))].
\]
In this light, the generator remains our structural model for the
predictive distribution, while the discriminator serves as a flexible,
jointly estimated auxiliary model. By continually adapting during
the training process, the discriminator forces the generator to match
an increasingly sophisticated set of data features, allowing for a
much richer and more flexible approximation of the predictive distribution
than classical methods allow.

\section{The Generative Predictive Distribution (GPD) method\label{sec:gpd}}

In this Section, we describe the proposed GPD methodology to approximate
the predictive distribution of a time series. The main ingredients
of GPD are a generative representation --- called \emph{generator}
in the GAN literature --- to be estimated, a \emph{discriminator}
that is used to compare the distribution of observed data to that
of simulated data generated by the model, and a minimax-type \emph{criterion
function}. Note that this is somewhat comparable to familiar extremum
estimation: a model to be estimated (now called the generator) and
a criterion function to be optimized (now in a minimax sense). The
main difference compared to extremum estimation is that the criterion
function for the generator is not readily available and must be approximated
using an auxiliary model (now called the discriminator). In this sense,
the procedure also shares an important similarity with simulation-based
estimation and indirect inference. The method we present below, which
involves a conditional GAN with MLP architecture and a relativistic
criterion function (that extends the work of \citealp{Jolicoeur:2019}),
is to the best of our knowledge new to the literature, albeit closely
related to the methods introduced in \citet{Zhou:etal:2023} and \citet{song:etal:202X}.

We proceed by defining each ingredient (namely, generator, discriminator,
and the criterion function) separately. In what follows, let $\left\{ Y_{t}\right\} _{t=1}^{n}$
be an observed time series of interest with dimension $d_{Y}$. At
each time $t$, a $d_{X}$-dimensional vector of conditioning variables
$X_{t}$, which may include lagged values of $Y_{t}$, is available
to the practitioner. For example, the analogue of a linear Gaussian
vector autoregressive process of order $p$ (VAR($p$)) is obtained
by letting $X_{t}=(1,Y{}_{t-1},\dots,Y{}_{t-p})$.\footnote{For the sake of uncluttered notation, in what follows we allow ourselves
to write $a=(a_{1},...,a_{n})$ for the (column) vector $a$ where
the components $a_{i}$ maybe either scalars or vectors (or both).} Assumptions regarding $Y_{t}$ and $X_{t}$ will be given below. 

\subsection{Generator\label{subsec:generator}}

We are interested in the predictive distribution of $Y_{t}$ given
$X_{t}$. In a setting where the pairs $(Y_{t},X_{t})$ are IID, it
was recently pointed out by \citet{Zhou:etal:2023} (see also \citealp{song:etal:202X})
that this predictive distribution could be represented using a certain
generative formulation. This representation is based on a folklore
result from measure theoretic probability going under various names
such as ``transfer theorem'' (\citealp[Thm 8.17]{Kallenberg:2021})
or ``noise outsourcing lemma'' (\citealp[Lemma 3.1]{Austin:2015})
and dating back to at least \citet{Aldous:1981}. A formulation suitable
for the present time series setting is as follows.

We introduce some notation following \citet[Ch 8]{Kallenberg:2021}.
Let $(\Omega,\mathcal{A},P)$ be a probability space. Suppose $Y_{t}$
and $X_{t}$ take values in $\mathcal{Y}\subseteq\mathbb{R}^{d_{Y}}$
and $\mathcal{X}\subseteq\mathbb{R}^{d_{X}}$ with these spaces being
endowed with their Borel $\sigma$-algebras $\mathcal{B}(\mathcal{Y})$
and $\mathcal{B}(\mathcal{X})$ and with $Y_{t}$ and $X_{t}$ being
measurable mappings from $(\Omega,\mathcal{A})$ to their respective
state spaces. We define a conditional distribution of $Y_{t}$, given
$X_{t}$, as a random measure of the form
\[
\mu(X_{t},B)=P\left\{ Y_{t}\in B\mid X_{t}\right\} \quad\text{a.s.},\quad B\in\mathcal{B}(\mathcal{Y}),
\]
for a probability kernel $\mu:\mathcal{X}\to\mathcal{\mathcal{Y}}$.
For brevity we also denote $\mu(x,\cdot)=P_{Y_{t}\mid X_{t}=x}$,
i.e., a measure in $\mathcal{B}(\mathcal{Y})$.
\begin{lem}
\noindent\label{lem:representation}Suppose $\left\{ Y_{t},X_{t}\right\} _{t\in\mathbb{Z}}$
is a strictly stationary stochastic process and $\left\{ Z_{t}\right\} _{t\in\mathbb{Z}}$
is an IID sequence of $d_{Z}$-dimensional $\mathcal{N}(0,I)$ random
vectors such that $Z_{t}$ is independent of $X_{t}$. Then there
exists a function $G:\mathbb{R}^{d_{Z}+d_{X}}\to\mathbb{R}^{d_{Y}}$
such that
\[
(G(Z_{t},X_{t}),X_{t})\overset{d}{=}(Y_{t},X_{t})\qquad\text{and}\qquad G(Z_{t},x)\sim P_{Y_{t}\mid X_{t}=x},\quad x\in\mathcal{X}.
\]
\end{lem}
Existing variants of this result (that we are aware of) are formulated
for random elements $(Y,X)$ without a time index, whereas above $(Y_{t},X_{t})$
does depend on $t$ but crucially the functional form of $G$ does
not. For completeness, a proof is given in Appendix \ref{subsec:representation}.

The key message of Lemma \ref{lem:representation} is that the predictive
distribution of $Y_{t}$ given $X_{t}$ can be expressed in the generative
representation $G(Z_{t},X_{t})$. For a given vector of conditioning
variables $X_{t}=x$, this generative representation is simply a function
mapping the innovation term $Z_{t}$ to the predictive distribution
$P_{Y_{t}\mid X_{t}=x}$. This explains the nomenclature \emph{generative}:
once the mapping $G(\cdot,\cdot)$ is available, draws from the distribution
$P_{Y_{t}\mid X_{t}=x}$ can then be obtained by simulating the Gaussian
$Z_{t}$ and calculating $G(Z_{t},x)$. In practice, the mapping $G(\cdot,\cdot)$
is of course unknown and must be estimated. Lemma \ref{lem:representation}
provides the rationale for approximating the predictive distribution
of $Y_{t}$ given $X_{t}$, or in other words the function $G(Z_{t},X_{t})$,
using a suitable parametric formulation. We note already now that
the function $G$ is highly non-unique and its parametric approximations
will not be identified --- we will return to this issue later in
Section \ref{subsec:Consistency}. Furthermore, note that Lemma \ref{lem:representation}
is quite general: It applies with $Y_{t}$ either uni- or multivariate,
with $X_{t}$ containing additional covariates or not, where $Y_{t}$
and $X_{t}$ can be continuous, discrete, or a combination of both,
and for either one-step or multiple-steps ahead predictive distributions.
A second key message of Lemma \ref{lem:representation} is that the
representation of the predictive distribution $G(Z_{t},X_{t})$ is
jointly distributed with $X_{t}$ as $P_{Y_{t},X_{t}}$. This suggests
that it is theoretically possible to learn about $G$ from the joint
distribution of $(Y_{t},X_{t})$. This is important, because while
the empirical version of $P_{Y_{t},X_{t}}$ is observed, in general
the empirical version of $P_{Y_{t}\mid X_{t}=x}$ is not.

One issue warranting a remark is the (potential) continuity of $G$.
If the joint distribution of $(Y_{t},X_{t})$ admits a strictly positive
and continuous density on a connected support, the existence of a
continuous generative map $G$ is guaranteed (via the multivariate
inverse function theorem applied to the Rosenblatt transformation,
see e.g. \citealp[Thm 8.18]{Davidson:1994} and \citealp{Rosenblatt:1952};
we illustrate the construction of such a continuous mapping in Appendix
\ref{subsec:continuity}). This continuity provides the theoretical
justification for approximating $G$ with neural networks, appealing
to standard universal approximation theorems \citep{hornik1989multilayer}. 

Another issue worth a remark is the choice and dimension of the distribution
of $Z_{t}$ ($d_{Z}$-dimensional Gaussian in Lemma \ref{lem:representation}).
The standard normal quantile function (applied element-wise) provides
a continuous mapping from standard uniform random variables to a multivariate
standard normal, suggesting that our construction does not crucially
depend on the choice of a specifically Gaussian reference distribution.
As for the dimension $d_{Z}$, although even univariate noise can
be continuously mapped to multivariate noise (using so-called space-filling
curves), such maps are unlikely to be well-behaving and some care
should be exercised when choosing the dimension $d_{Z}$. In our practical
applications, we use $d_{Z}=d_{Y}$.\textcolor{red}{{} }

With the generative representation of Lemma \ref{lem:representation}
in hand, we next approximate the mapping $G(\cdot,\cdot)$ therein
by the \emph{generator}, a parametric map $G_{\gamma}:\mathbb{R}^{d_{Z}+d_{X}}\to\mathbb{R}^{d_{Y}}$
indexed by $\gamma\in\Gamma\subset\mathbb{R}^{d_{\gamma}}$. Let $Z_{t}\overset{\mathrm{IID}}{\sim}\mathcal{N}\left(0,I\right)$
be a $d_{Z}-$dimensional innovation. We focus on a generator based
on the multilayer perceptron (MLP)
\begin{equation}
G_{\gamma}\left(Z_{t},X_{t}\right)=\gamma_{O}\circ a\circ\gamma_{L_{G}}\circ\cdots\circ a\circ\gamma_{1}\left(Z_{t},X_{t}\right).\label{eq:generator}
\end{equation}
Here the layers $\gamma_{1}:\mathbb{R}^{d_{Z}+d_{X}}\to\mathbb{R}^{H_{G}}$,
$\gamma_{\ell}:\mathbb{R}^{H_{G}}\to\mathbb{R}^{H_{G}}$, $\ell=2,\dots,L_{G}$,
and $\gamma_{O}:\mathbb{R}^{H_{G}}\to\mathbb{R}^{d_{Y}}$ are affine
functions, possibly including an intercept, where the number of layers
$L_{G}$ and the number of hidden units per layer $H_{G}$ are chosen
by the modeler. The pre-specified nonlinear function $a(\cdot)$ is
known as the activation function, and $a(\cdot)$ is applied element-wise
when the argument is a matrix. Our choice is $a(u)=\max\{u,0\}$,
often called the rectified linear unit, or ReLU for short. Other common
choices are $a(u)=1/(1+\exp(-u))$, also known as \emph{sigmoid},
and $a(u)=\tanh(u)$, though we do not cover these cases explicitly.
ReLU activations are the current industry standard and\textcolor{red}{{}
}have been shown\textcolor{red}{{} }to have attractive properties in
terms of their expressive power \citep{Yarotsky2017}. The \emph{output}
layer, $\gamma_{O}$, is a linear transformation of the hidden layer.
In line with the strict stationarity assumed in Lemma \ref{lem:representation},
this particular generator is only able to sample from a time-invariant
predictive distribution. This is less restrictive than it might appear
at first sight, because this still allows for rich but nevertheless
stationary dynamics in the conditional distribution (e.g., ARMA type
variation in the conditional mean, GARCH type variation in the conditional
variance, etc.). This is the same architecture used in \citet{Zhou:etal:2023},
with the slight difference that we allow $X_{t}$ to include lags
of $Y_{t}$.

\subsection{Discriminator\label{subsec:discriminator}}

The second key component of the GPD framework is the discriminator
function, which plays the role of an auxiliary model used to compare
the distribution of observed data to that of simulated data generated
by the model. Formally, the discriminator is a measurable function
$D_{\delta}:\mathbb{R}^{d_{Y}+d_{X}}\to\mathbb{R}$, indexed by a
finite-dimensional parameter vector $\delta\in\Delta$, and evaluated
at pairs $\left(Y_{t},X_{t}\right)$. In this paper, the discriminator
is specified as a multilayer perceptron (MLP) with ReLU activation.
Let $H_{D}$ denote the number of hidden units per layer and $L_{D}$
the number of hidden layers. Specifically,
\begin{equation}
D_{\delta}\left(Y_{t},X_{t}\right)=\delta_{O}\circ a\circ\delta_{L_{D}}\circ\cdots\circ a\circ\delta_{1}\left(Y_{t},X_{t}\right),\label{eq:discriminator}
\end{equation}
where $\delta_{1}:\mathbb{R}^{d_{Y}+d_{X}}\to\mathbb{R}^{H_{D}}$,
$\delta_{\ell}:\mathbb{R}^{H_{D}}\to\mathbb{R}^{H_{D}}$, $\ell=2,\dots,L_{D}$
and $\delta_{O}:\mathbb{R}^{H_{D}}\to\mathbb{R}$ are linear maps,
possibly including an intercept. Again, only the parameters in the
linear maps are estimated. This is also the same architecture employed
by \citet{Zhou:etal:2023}.

From an econometric perspective, the discriminator defines a class
of test functions used to assess the adequacy of the generator. When
this class is sufficiently rich, the minimax problem (to be formally
defined in the next subsection) forces the generator-induced conditional
distribution to match the true conditional distribution along a wide
range of features, including nonlinear and higher-order characteristics.
In population, the discriminator can detect any systematic discrepancy
between the observed and simulated conditional distributions that
lies within its function class.

The MLP specification is particularly convenient for this role. Neural
networks with standard non-polynomial activation functions are universal
approximators \citep{hornik1989multilayer}, allowing the discriminator
to represent complex, non-linear functions of $(Y_{t},X_{t})$ that
characterize the underlying data distribution. Crucially, for a fixed
architecture (fixed number of layers and units per layer), the resulting
function class possesses a finite Vapnik-Chervonenkis (VC) dimension
\citep{Bartlett:etal:2019}. This bounded complexity is essential
for establishing uniform convergence results; it ensures that the
discriminator's empirical performance is close to the best performance
in its class as the sample size grows. Furthermore, this framework
allows for the derivation of consistency results even under conditions
of weak dependence, such as $\beta$-mixing processes, which are common
in time series applications.

It is important to emphasize that the discriminator is not interpreted
as a probabilistic classifier. Instead, it should be viewed as an
auxiliary device that adaptively selects informative features of the
data against which the generator is evaluated. In this sense, the
GPD framework generalizes classical simulation-based estimation methods.
For example, if the discriminator is restricted to be linear in $(Y_{t},X_{t})$,
the minimax problem reduces to a form of conditional moment matching.
Allowing for nonlinear discriminators extends this idea by endogenously
selecting moments that are most informative for distinguishing observed
and simulated data. While the discriminator is essential for estimation,
it plays no direct role once the generator has been estimated. All
subsequent analysis, such as simulation, forecasting, and computation
of functionals of the predictive distribution, is conducted using
the generator alone.

\subsection{Criterion function\label{subsec:criterion}}

Estimation of the generator and the discriminator is based on a minimax
problem whose goal is to implicitly compare the model-generated data
to observed data. Based on Lemma \ref{lem:representation}, this amounts
to finding a generator $G$ so that the distribution of $(G(Z_{t},X_{t}),X_{t})$
matches that of $(Y_{t},X_{t})$. The population criterion function
we employ is defined as
\begin{equation}
f(\gamma,\delta)=\mathbb{E}\bigl[\ln\sigma\bigl(D_{\delta}(Y_{t},X_{t})-D_{\delta}(G_{\gamma}(Z_{t},X_{t}),X_{t})\bigr)\bigr]+\ln2,\label{eq:f}
\end{equation}
where $\sigma(u)=1/(1+\exp(-u))$ is the sigmoid function and the
expectation is taken with respect to the joint distribution of $\left(X_{t},Y_{t},Z_{t}\right)$.
The population game is formulated as follows:
\begin{equation}
\adjustlimits\inf_{\gamma\in\Gamma}\sup_{\delta\in\Delta}\,f(\gamma,\delta).\label{eq:population minimax problem}
\end{equation}
This formulation defines a two-player sequential game between the
generator and the discriminator. For a fixed generator parameter $\gamma$,
the discriminator chooses $\delta$ to maximize its ability to distinguish
observed outcomes from simulated ones. The generator, who is given
first-mover advantage, anticipates the move of the discriminator and
chooses $\gamma$ to minimize its worst-case scenario. We note that
in the traditional convex-concave setting ($f(\cdot,\delta)$ convex
for all fixed $\delta\in\Delta$ and $f(\gamma,\cdot)$ concave for
all fixed $\gamma\in\Gamma$), the classical von Neumann minimax theorem
implies that $\inf_{\gamma\in\Gamma}\sup_{\delta\in\Delta}f(\gamma,\delta)=\sup_{\delta\in\Delta}\inf_{\gamma\in\Gamma}f(\gamma,\delta)$
under mild conditions. In contrast to this, \citet{jin2020what} and
others have emphasized that in most modern machine learning applications
$f$ is non-convex and non-concave and the order in which minimization
and maximization are performed matters.

Our choice of the specific functional form of the criterion functions
is somewhat different from those used by \citet{Zhou:etal:2023} and
\citet{song:etal:202X}. \citet{Zhou:etal:2023} use a criterion based
on the Kullback-Leibler divergence which can (in our notation) be
expressed as
\begin{equation}
f_{\text{Zhou}}\left(\gamma,\delta\right)=\mathbb{E}[D_{\delta}(G_{\gamma}(Z_{t},X_{t}),X_{t})]-\mathbb{E}[\exp(D_{\delta}(Y_{t},X_{t}))].\label{eq:f_Zhou}
\end{equation}
\citet{song:etal:202X} use a criterion based on the 1-Wasserstein
distance given by
\begin{equation}
f_{\text{Song}}\left(\gamma,\delta\right)=\mathbb{E}[D_{\delta}(G_{\gamma}(Z_{t},X_{t}),X_{t})]-\mathbb{E}[D_{\delta}(Y_{t},X_{t})],\label{eq:f_Song}
\end{equation}
and combine this with a regularization term based on the least squares
loss (see their Section 2.2 for details). It is clear that criteria
(\ref{eq:f}), (\ref{eq:f_Zhou}), and (\ref{eq:f_Song}) all measure
discrepancies between model-generated outcomes and observed outcomes
but in slightly different ways. Intuitively, the discriminator assigns
higher values to realizations that are more likely to originate from
the true data-generating process than from the generator distribution.
Many other formulations of the adversarial criterion would of course
also be possible.

The specific functional form of the criterion (\ref{eq:f}) we employ
is motivated by a number of considerations. First, GAN-type problems
are typically notoriously difficult to solve due to their adversarial
nature. The so-called relativistic GAN criterion (\ref{eq:f}) was
proposed by Jolicoeur-Martineau \citeyearpar{Jolicoeur:2019,Jolicoeur:2020}
with the aim of improving stability of numerical optimization, and
further advocated by \citet{Huang:etal:2024} as a well-behaved criterion
that is amenable to optimization without the need of several ad-hoc
tricks commonly employed in GAN estimation problems. We have chosen
the criterion of form (\ref{eq:f}) mainly due to our aim of proposing
a general method that works with computational ease with a criterion
that requires no choice of tuning parameters, at least in our applications.\footnote{\citet{Huang:etal:2024} document divergence of the training procedure
of relativistic GAN in an application to handwritten digit recognition.
That exercise concerns a high-dimensional distribution with 1,000
modes, which is fundamentally different in nature from the applications
considered here. While in our applications we found no need for introducing
additional regularization terms, we note that in large-scale applications
the use of regularization is probably unavoidable.} Another reason for this choice is that it facilitates our theoretical
analysis. In the next section, we will see that using criterion (\ref{eq:f})
leads to consistent estimation of the generator and discriminator
parameters. Furthermore, we show below that this criterion, when maximized
with respect to $D$, corresponds to the Jensen-Shannon (JS) distance
between two probability measures involving the real and fake predictive
distributions. Lemma \ref{lem:criterion} below, proven in Appendix
\ref{subsec:criterionproof}, formalizes this intuition. 

Before stating the lemma, we introduce some additional notation. Suppose
here that $Y$, $X$, and $Z$ are any random variables with $Z$
being independent of $Y$ and $X$. Let $P_{Y,X}$, $P_{X}$ and $P_{Z}$
be the distributions of $(Y,X)$, $X,$ and $Z$, respectively, where
$Z$ takes values in $\mathcal{Z}\subseteq\mathbb{R}^{d_{Z}}$. For
any measurable function $\phi$, let
\begin{align*}
\mathbb{E}[\phi(Y,G(X,Z),X)] & =\int_{\mathcal{X}}\int_{\mathcal{Y}}\int_{\mathcal{Z}}\phi(y,G(z,x),x)P_{Y,X}(dy,dx)P_{Z}(dz)\\
 & =\int_{\mathcal{X}}\int_{\mathcal{Y}}\int_{\mathcal{Y}}\phi(y,g,x)\mu(x,dy)\mu_{G}(x,dg)P_{X}(dx),
\end{align*}
where $\mu(x,\cdot)=P_{Y\mid X=x}$, and $\mu_{G}(x,\cdot)=Q_{G\mid X=x}$
is the push-forward measure of $P_{Z}$ under the measurable mapping
$z\mapsto G(z,x)$. For two generic probability measures $P$ and
$Q$ defined on the same measurable space $\mathcal{S}$, with $P$
absolutely continuous with respect to $Q$, the Kullback-Leibler divergence
is defined as $\mathrm{KL}\left(P\parallel Q\right)=\int_{s\in\mathcal{S}}\ln\frac{dP}{dQ}(s)P(ds),$
where $\frac{dP}{dQ}$ is the Radon-Nikodym derivative of $P$ with
respect to $Q$. Moreover, the Jensen-Shannon distance is defined
as
\[
\mathrm{JS}\left(P\parallel Q\right)=\frac{1}{2}\mathrm{KL}\left(P\parallel\frac{1}{2}(P+Q)\right)+\frac{1}{2}\mathrm{KL}\left(Q\parallel\frac{1}{2}(P+Q)\right).
\]
We can now state the following result.
\begin{lem}
\label{lem:criterion}Let $\mathcal{L}(G,D)=\mathbb{E}[\ln\sigma(D(Y,X)-D(G(Z,X),X))]+\ln2.$
It holds that
\begin{equation}
\sup_{D}\mathcal{L}(G,D)=\mathrm{JS}\left(P_{1}\parallel P_{2}\right),\label{eq:JS}
\end{equation}
where the supremum over $D$ is taken over the class of measurable
functions, $P_{1}=P_{X}\otimes\mu\otimes\mu_{G}$ and $P_{2}=P_{X}\otimes\mu_{G}\otimes\mu$
are probability measures defined by $P_{1}(dx,dy,dg)=P_{X}(dx)\mu(x,dy)\mu_{G}(x,dg)$
and $P_{2}(dx,dy,dg)=P_{X}(dx)\mu_{G}(x,dy)\mu(x,dg).$ Moreover,
the supremum is attained by any function $D^{*}$ such that
\begin{equation}
D^{*}(y,x)=\ln\frac{d\mu(x,\cdot)}{d\nu(x,\cdot)}(y)-\ln\frac{d\mu_{G}(x,\cdot)}{d\nu(x,\cdot)}(y)-C(x),\label{eq:D*}
\end{equation}
where $\nu(x,\cdot)=\frac{1}{2}[\mu(x,\cdot)+\mu_{G}(x,\cdot)]$ and
$C(x)$ is a measurable function dependent only on $x$.
\end{lem}
In Lemma \ref{lem:criterion}, $\mathcal{L}(G,D)$ coincides with
the population criterion in \eqref{eq:f}, whenever $G$ and $D$
belong to the function classes defined in \eqref{eq:generator} and
\eqref{eq:discriminator} and $(Y,X,Z)=(Y_{t},X_{t},Z_{t})$. In \citet{Goodfellow:etal:2014},
it is shown that the supremum over $D$ of their population criterion
is (up to constants) the JS distance between real and fake data distributions.
In our case, we also obtain a JS distance, but the real ($\mu$) and
fake ($\mu_{G}$) distributions make an appearance through $P_{1}$
and $P_{2}$. Concretely, $P_{1}$ is the joint distribution of $(X,Y,G(X,Z))$,
while $P_{2}$ is a measure in which, compared to $P_{1}$, the roles
of $\mu=P_{Y\mid X=x}$ and $\mu_{G}=Q_{G\mid X=x}$ are reversed.
Clearly, when $\mu=\mu_{G}$, then $\mathrm{JS}(P_{1}\parallel P_{2})=0$,
and the converse also holds: if $\mathrm{JS}(P_{1}\parallel P_{2})=0$,
then $P_{1}=P_{2}$, which implies that $\mu=\mu_{G}$, ($P_{X}-$almost
surely). The function $D^{*}$, which is not uniquely defined, can
be regarded as a log-density ratio between $\mu$ and $\mu_{G}$,
and the form of $C(x)$ is irrelevant for our purposes as it cancels
out upon taking differences. This specific formulation allows both
the data and generator distributions to be either continuous, discrete,
or a combination of both, and avoids requiring $\mu$ to be absolutely
continuous with respect to $\mu_{G}$, and vice versa. This is particularly
relevant in our case, where in general a ReLU-MLP generator's push-forward
distribution can have atoms, ruling out the existence of a density
with respect to the Lebesgue measure.

\section{Estimation\label{sec:Estimation}}

\subsection{Practical estimation}

In real applications, we would have an observed time series $\left\{ Y_{t}\right\} _{t=1}^{n}$
as well as conditioning variables $\left\{ X_{t}\right\} _{t=1}^{n}$
available to the practitioner. (When $X_{t}$ contains lagged values
of $Y_{t}$, we assume appropriate initial values, such as $Y_{0},\dots,Y_{-p}$
in the case of a vector autoregressive process of order $p$, are
also available.) The sample analogue of the population criterion function
\eqref{eq:f} is
\begin{align*}
\widehat{f}_{n}\left(\gamma,\delta\right) & =\frac{1}{n}\sum_{t=1}^{n}\ln\sigma\bigl(D_{\delta}(Y_{t},X_{t})-D_{\delta}(G_{\gamma}(Z_{t},X_{t}),X_{t})\bigr),
\end{align*}
where $\left\{ Z_{t}\right\} _{t=1}^{n}$ is an IID sequence of latent
innovations independent of the observed data. In our applications,
we choose the normal distribution $\mathcal{N}(0,I)$ with the same
dimension as the dimension of $Y$. However, as we argued above, in
this setup normality of $Z$ is not restrictive, and other choices
are possible. Estimation is based on the approximate solutions to
the corresponding empirical version of the game \eqref{eq:population minimax problem}\textcolor{red}{{}
}given by
\begin{equation}
\min_{\gamma\in\Gamma}\max_{\delta\in\Delta}\,\widehat{f}_{n}\left(\gamma,\delta\right).\label{eq:sample minimax problem}
\end{equation}

Finding a solution, let alone all solutions, to this minimax game
can be a challenging task and a large body of machine learning literature
focuses on this (see, e.g., \citealp{diakonikolas2021efficient},
\citealp{fiez2021}, \citealp{mangoubi2021}, and the references therein).
In this paper, we have carefully chosen the generator, discriminator,
and the criterion function to be rich enough to yield interesting
empirical results in the three applications we present in Section
\ref{sec:empirical}, while at the same time achieving relative computational
ease. Details of the specific algorithm we use are provided in Algorithm
\ref{algo:gpd}.\footnote{It should be remarked that in a large part of the GAN literature there
is a fundamental mismatch between how the problem is formulated (minimax)
and how it is actually solved in practice (alternating min-max, potentially
non-zero sum game). We follow the same practice here, and regard the
practical solution found by Algorithm 1 as a computational trick to
obtain an approximate solution to \eqref{eq:sample minimax problem}.} In this algorithm, we follow the standard machine learning practice
of approximately solving both optimization problems using a variant
of \emph{stochastic gradient descent}. For clarity, we briefly define
the specific jargon used in this literature. The qualifier \emph{stochastic}
indicates that instead of evaluating the average gradient of the criterion
function using all observations, we evaluate it only on a random subset
(or \emph{batch}). We use this stochastic gradient to update the parameters
via the so-called \emph{Adam} optimizer, an updating rule based on
the exponentially weighted moving averages of the gradient\textquoteright s
first and second moments. This rule relies on several hyperparameters,
the most critical being the \emph{learning rate}, which dictates the
step size of the updates. Parameters are updated batch by batch until
the entire sample has been processed --- a full cycle referred to
as an \emph{epoch}. Typically, this procedure continues until a stopping
criterion is met (\emph{early stopping}). For further details on the
Adam optimizer, see \citet{Kingma:Ba:2015}.

The primary modification in Algorithm \ref{algo:gpd}, compared to
standard Adam-based optimization, is that we update the discriminator
and generator in an alternating fashion. For each batch, we first
execute an Adam-based Discriminator step (D-step), followed immediately
by an Adam-based Generator step (G-step). Beyond the inherent complexity
of solving a minimax rather than a standard optimization problem,
this setup introduces an additional challenge: the criterion function
is relative to each player. Consequently, defining a reliable metric
for early stopping becomes a non-trivial task. One practical suggestion
which we follow is to use the Wasserstein distance based on random
projections also known as sliced Wasserstein's distance, or SWD for
short (see Appendix \ref{subsec:swd} for further details).

\begin{algorithm}[tb]
\begin{algorithmic}[1]
\Require Initialize $\gamma$ and $\delta$. Choose Adam optimizer parameters, batch size $B$ and stopping criterion.
\State Split the data in batches of size $B$, i.e. $\left\{Y_b^{(1)},X_b^{(1)}\right\}_{b=1}^B, \ldots,\left\{Y_b^{(n/B)},X_b^{(n/B)}\right\}_{b=1}^B$.
\While{not converged}
\For{$i=1,\dots,n/B$}
\State \textbf{D-step:}
\State \hspace{1em} Draw $\left\{ Z_{b}^{(i)}\right\} _{b=1}^{B}\overset{\text{IID}}{\sim}\mathcal{N}\left(0,I\right)$.
\State \hspace{1em} Compute $\left\{ D_{\delta}(Y_{b}^{(i)},X_{b}^{(i)})\right\} _{b=1}^{B}$ and $\left\{ D_{\delta}(G_{\gamma}(X_{b}^{(i)},Z_{b}^{(i)}),X_{b}^{(i)})\right\} _{b=1}^{B}$.
\State \hspace{1em} Compute discriminator loss
$$
L_{D}(\delta)=-\frac{1}{B}\sum_{b=1}^{B}\ln\sigma\left(
D_{\delta}(Y_{b}^{(i)},X_{b}^{(i)})
-D_{\delta}(G_{\gamma}(X_{\pi(b)}^{(i)},Z_{\pi(b)}^{(i)}),X_{\pi(b)}^{(i)})
\right),
$$
\hspace{4em} where $\pi(\cdot)$ is the random permutation operator.
\State \hspace{1em} Take an Adam step to update $\delta$.
\State \textbf{G-step:}
\State \hspace{1em} Re-draw $\left\{ Z_{b}^{(i)}\right\} _{b=1}^{B}\overset{\text{IID}}{\sim}\mathcal{N}\left(0,I\right)$.
\State \hspace{1em} Compute $\left\{ D_{\delta}(G_{\gamma}(X_{b}^{(i)},Z_{b}^{(i)}),X_{b}^{(i)})\right\} _{b=1}^{B}$.
\State \hspace{1em} Compute generator loss
$$
L_{G}(\delta)=-\frac{1}{B}\sum_{b=1}^{B}\ln\sigma\left(
D_{\delta}(G_{\gamma}(X_{b}^{(i)},Z_{b}^{(i)}),X_{b}^{(i)})
-D_{\delta}(Y_{\pi(b)}^{(i)},X_{\pi(b)}^{(i)})
\right).
$$
\State \hspace{1em} Take an Adam step to update $\gamma$.
\EndFor
\EndWhile
\end{algorithmic}

\caption{}

\label{algo:gpd}
\end{algorithm}

Studying algorithmic convergence and stability of Algorithm \ref{algo:gpd}
would be interesting but is well beyond the scope of this paper. The
key reason for this is that our population game is inherently nonconvex-nonconcave
in its arguments. Nonconvex-nonconcave minimax optimization problems
have gained widespread interest in the machine learning community
in recent years but are notoriously difficult to study. We leave issues
of algorithmic convergence to future research and direct interested
readers to the recent papers of \citet{li2025nonsmooth} and \citet{lin2025two}
for further details and discussion on this topic.

\subsection{Consistency\label{subsec:Consistency}}

We next turn to statistical properties of the solutions of the empirical
version of the problem (\ref{eq:sample minimax problem}). In particular,
we study under which conditions these solutions, say $\hat{\gamma}_{n}$
and $\hat{\delta}_{n}$, are consistent estimators of the corresponding
solutions, say $\gamma_{0}$ and $\delta_{0}$, of the population
minimax problem (\ref{eq:population minimax problem}). Consistency
results for the solutions of minimax problems have previously been
studied only in an IID setting, see e.g. \citet[Theorem 5.9]{Shapiro:etal:2021}
and \citet[Theorem 1]{Meitz:2024}. The developments below extend
these previous results to our current time series setting.

We begin by introducing necessary notation and assumptions. Our first
assumption is a suitable modification of Assumption GAN in \citet{Meitz:2024}.
In this assumption and in what follows, we use either notation $\theta$
or notation $(\gamma,\delta)$ for the elements of the set $\Theta=\Gamma\times\Delta$.
We use $\left\Vert \cdot\right\Vert $ to denote Euclidean distance.
\begin{assumption}
\label{asm:gpd}(Assumption GPD).
\end{assumption}
\begin{enumerate}[label=\theassumption(\alph{enumi})]
\item \label{asm:data}\textcolor{teal}{$\left\{ Y_{t},X_{t}\right\} _{t\in\mathbb{Z}}$}
is a strictly stationary $\beta-$mixing stochastic process with the
same distribution as $\left(Y,X\right)$ and taking values in $\mathcal{Y\times\mathcal{X}},$
with $\mathcal{Y}\subseteq\mathbb{R}^{d_{Y}}$ and $\mathcal{X}\subseteq\mathbb{R}^{d_{X}}$.
Moreover, the $\beta-$mixing coefficients $\beta_{k}$ satisfy
\[
k^{p/(p-2)}(\ln k)^{2(p-1)/(p-2)}\beta_{k}\to0\quad\text{as}\quad k\to\infty\quad\text{for some}\quad2<p<\infty.
\]
\item \label{asm:Z}\textcolor{teal}{$\left\{ Z_{t}\right\} _{t\in\mathbb{Z}}$}
is an IID sequence of random vectors with the same distribution as
$Z$ and taking values in $\mathcal{Z}\subseteq\mathbb{R}^{d_{Z}}$.
\item \label{asm:compact}The set $\Theta=\Gamma\times\Delta\subseteq\mathbb{R}^{d_{\gamma}+d_{\delta}}$
is compact and nonempty.
\item \label{asm:criterion}The generator function $G_{\gamma}$ defined
in \eqref{eq:generator}, the discriminator function $D_{\delta}$
defined in \eqref{eq:discriminator}, and the criterion function $F$
defined by
\begin{align*}
F & (x,y,z,\theta)=\ln\sigma\bigl(D_{\delta}(y,x)-D_{\delta}(G_{\gamma}(z,x),x)\bigr)
\end{align*}
where $\sigma(u)=1/(1+\exp(-u))$, are such that $F(X,Y,Z,\theta)$
is measurable for all $\theta\in\Theta$ and continuous on $\Theta$
with probability one. Moreover, the activation function appearing
in \eqref{eq:generator} and \eqref{eq:discriminator} is $a(u)=\max\{0,u\}$.
\item \label{asm:moments}We have $\mathbb{E}[\left\Vert X_{t}\right\Vert ^{p}]<\infty$,
$\mathbb{E}[\left\Vert Y_{t}\right\Vert ^{p}]<\infty$, and $\mathbb{E}[\left\Vert Z_{t}\right\Vert ^{p}]<\infty$.
\end{enumerate}
Compared to Assumption GAN of \citet{Meitz:2024}, Assumptions \ref{asm:Z}
and \ref{asm:compact} are identical. In contrast, \ref{asm:data}
allows the data to exhibit temporal dependence and imposes a time-invariant
law of $Y_{t}$ given $X_{t}$. Regarding the $\beta-$mixing coefficients,
which quantify how quickly the ``memory'' of the time series decays,
the polynomial decay condition assumed in \ref{asm:data} is sufficiently
mild to accommodate many time series processes encountered in practice
(see, e.g., \citealp{Carrasco:Chen:2002}, \citealp{Meitz:Saikkonen:2008a,Meitz:Saikkonen:2008b},
and the references therein). Assumption \ref{asm:criterion} is sufficient
to ensure three main requirements (in the proof of Theorem 1 below):
(i) that the supremum over $\Delta$ of the criterion is continuous;
(ii) that the criterion function is of sufficiently controlled complexity,
meaning that it belongs to a VC-subgraph class of functions \citep[Section 2.6]{vanderVaart:Wellner:2023};
and (iii) the existence of an envelope function for the class of functions
$F$ parameterized by $\theta\in\Theta$. In our applications, the
relativistic criterion leads to more stable estimation compared to
Wasserstein distance and binary cross-entropy. We focus on the ReLU
activation function to keep proofs as simple as possible, but our
results can be generalized to cover other commonly used activation
functions such as sigmoid or tanh. Assumption \ref{asm:moments} ensures
the existence of $p$ moments of the envelope function of $F$, which
is standard in empirical process theory. We also remark that although
alternative conditions involving the quantile function of the envelope
could be employed to achieve similar results for slightly more general
functions \citep{rio2017asymptotic}, the VC-subgraph property is
general enough for our purposes. Finally, we refer to \citet{Arcones:Yu:1994}
for more detailed definitions of the $\beta$--mixing coefficients
$\beta_{k}$ and the envelope function.

The goal is to estimate $\Theta_{0}$, the set of all solutions to
the population game (\ref{eq:population minimax problem}). This set
can be expressed as
\[
\Theta_{0}=\left\{ \left(\gamma_{0},\delta_{0}\right)\in\Theta:f\left(\gamma_{0},\delta_{0}\right)=\sup_{\delta\in\Delta}f\left(\gamma_{0},\delta\right)=\varphi\left(\gamma_{0}\right),\varphi\left(\gamma_{0}\right)=\inf_{\gamma\in\Gamma}\varphi\left(\gamma\right)=V_{0}\right\} ,
\]
where $\varphi(\gamma)=\sup_{\delta\in\Delta}f(\gamma,\delta)$. As
noted in \citet{Meitz:2024}, we may think of $\Theta_{0}$ as the
level set of minimizers of the population criterion function
\[
Q\left(\theta\right)=\varphi\left(\gamma\right)-\min\left\{ f\left(\gamma,\delta\right),V_{0}\right\} .
\]
This criterion has a minimum of zero, so we write $\Theta_{0}=\left\{ \theta\in\Theta:Q\left(\theta\right)=0\right\} .$
Sample analogues are obtained by replacing expectations with empirical
averages, i.e. 
\[
\widehat{f}_{n}(\theta)=\frac{1}{n}\sum_{t=1}^{n}F(X_{t},Y_{t},Z_{t},\theta),\quad\widehat{Q}_{n}(\theta)=\widehat{\varphi}_{n}(\gamma)-\min\left\{ \hat{f}_{n}(\gamma,\delta),\widehat{V}_{0}\right\} ,
\]
where $\widehat{\varphi}_{n}(\gamma)=\sup_{\delta\in\Delta}\hat{f}_{n}(\gamma,\delta)$
and $\widehat{V}_{0}=\inf_{\gamma\in\Gamma}\sup_{\delta\in\Delta}\hat{f}_{n}(\gamma,\delta)$,
so that the set of (exact) solutions to the empirical minimax problem
(\ref{eq:sample minimax problem}) can be expressed as $\widehat{\Theta}_{n}=\{\theta\in\Theta:\widehat{Q}_{n}\left(\theta\right)=0\}$.
For our theoretical developments, we assume that a suitable algorithm
to estimate $\Theta_{0}$ is available to the researcher (e.g., our
Algorithm \ref{algo:gpd} or any alternative algorithm). A sequence
of non-negative random variables $\tau_{n}\overset{p}{\rightarrow}0$
is introduced to represent the slackness between the exact solution
$\widehat{\Theta}_{n}$ and the approximate solution found by the
algorithm, namely
\[
\widehat{\Theta}_{n}(\tau_{n})=\{\theta\in\Theta:\widehat{Q}_{n}(\theta)\leq\tau_{n}\}.
\]

Some discussion is warranted. The population minimax problem (\ref{eq:population minimax problem})
is likely to have numerous solutions and thus $\Theta_{0}$ is set-valued
and potentially not finite. Key reasons for this are the non-uniqueness
of the function $G$ of Lemma \ref{lem:representation} representing
the predictive distribution of $Y_{t}$ given $X_{t}$ as well as
the inherent non-identifiability of the MLPs used in the generator
and the discriminator. As for the interpretation of $\Theta_{0}$,
the MLP $G_{\gamma}(z,x)$ is a parametric approximation to $G(z,x),$
the generative representation of the predictive distribution of $Y_{t}$
given $X_{t}$, but we do not assume that $G_{\gamma}(z,x)$ would
for some $\gamma$ be equal to $G(z,x)$. In this sense, the elements
of $\Theta_{0}$ do not correspond to any ``true'' parameter values.
Furthermore, as the parameters $\gamma$ and $\delta$ themselves
carry no essential interpretation and as all elements of the set $\Theta_{0}$
are equally valid for approximating the generative representation,
this multitude of solutions is mainly an issue complicating numerical
estimation. Indeed, $\widehat{\Theta}_{n}(\tau_{n})$ is likely to
be set-valued and potentially large. An estimation algorithm searching
for approximate solutions to the empirical minimax problem would typically
return a single solution, say $\hat{\theta}_{n}(\tau_{n})$, and several
runs of the algorithm with different initializations would be required
to find the entire set $\widehat{\Theta}_{n}(\tau_{n})$. 

Now consider consistency. Our goal is to prove that an appropriately
defined distance between the sets $\widehat{\Theta}_{n}(\tau_{n})$
and $\Theta_{0}$ converges to zero in probability. We define the
Hausdorff distance between two sets $A$ and $B$ as
\[
d_{H}\left(A,B\right)=\max\left\{ \sup_{a\in A}d\left(a,B\right),\sup_{b\in B}d\left(b,A\right)\right\} ,\quad\text{with}\quad d\left(a,B\right)=\inf_{b\in B}\left\Vert a-b\right\Vert .
\]
Since $d_{H}\left(\widehat{\Theta}_{n}\left(\tau_{n}\right),\Theta_{0}\right)\overset{p}{\rightarrow}0$
follows from the two conditions 

\begin{equation}
\text{(a)}\:\sup_{\theta\in\widehat{\Theta}_{n}\left(\tau_{n}\right)}d\left(\theta,\Theta_{0}\right)\overset{p}{\rightarrow}0\qquad\text{and}\qquad\text{(b)}\:\sup_{\theta\in\Theta_{0}}d\left(\theta,\widehat{\Theta}_{n}\left(\tau_{n}\right)\right)\overset{p}{\rightarrow}0,\label{eq:onesided-hsdrf}
\end{equation}
it suffices to establish \eqref{eq:onesided-hsdrf} (a) and (b). Intuitively,
(a) means that $\widehat{\Theta}_{n}(\tau_{n})$ is not too large
compared to $\Theta_{0}$, whereas (b) means that $\widehat{\Theta}_{n}(\tau_{n})$
is large enough to cover all of $\Theta_{0}$. 

We can now state our main consistency result (the proof is given in
Appendix \ref{subsec:proof-thm}).
\begin{thm}
\label{thm:consistent}Suppose Assumption \ref{asm:gpd} holds.
\end{thm}
\begin{enumerate}[label=\thethm(\alph{enumi})]
\item \label{thm:1a}If $\tau_{n}$ is a sequence of non-negative random
variables such that $\tau_{n}\overset{p}{\rightarrow}0$, then condition
(\ref{eq:onesided-hsdrf})~(a) holds.
\item \label{thm:1b} If $\tau_{n}$ is a sequence of positive random variables
such that $\tau_{n}\overset{p}{\rightarrow}0$ and $n^{-1/2}/\tau_{n}\overset{p}{\rightarrow}0$,
then condition (\ref{eq:onesided-hsdrf}) (b) also holds so that $d_{H}\left(\widehat{\Theta}_{n}\left(\tau_{n}\right),\Theta_{0}\right)\overset{p}{\rightarrow}0$.
\end{enumerate}
Regarding the assumptions on the slackness sequence $\tau_{n}$, note
that exact solutions ($\tau_{n}=0$) are allowed in part (a) of this
theorem, but the stronger requirement in part (b) requires the slackness
sequence to converge to zero at a rate slower than $1/\sqrt{n}$.
Our Algorithm \ref{algo:gpd} used for finding an approximate solution
to the empirical minimax problem alternates updates of $\gamma$ and
$\delta$ over multiple epochs (i.e. full passes to the data). The
total number of epochs is selected by monitoring the trajectory of
some suitably chosen metric (e.g. the sliced Wasserstein distance
presented in Appendix \ref{subsec:swd}), and training is stopped
once further epochs yield no appreciable change in the stopping criterion.
The tolerance associated with declaring convergence to a solution
should intuitively decrease with the sample size, consistent with
a choice of slackness sequence like $\tau_{n}=n^{-0.49}$.

As for the conclusions of the Theorem, the result in part (b) shows
that the set of approximate solutions to the empirical minimax problem
is a Hausdorff consistent estimator of $\Theta_{0}$, regardless of
the number of solutions. Note that this Hausdorff consistency result
is for the entire set $\widehat{\Theta}_{n}(\tau_{n})$ --- a single
element, say $\hat{\theta}_{n}$, of $\widehat{\Theta}_{n}(\tau_{n})$
would only satisfy one-sided consistency $d(\hat{\theta}_{n},\Theta_{0})\overset{p}{\rightarrow}0$
but obviously not $\sup_{\theta\in\Theta_{0}}d(\theta,\hat{\theta}_{n})\overset{p}{\rightarrow}0$
when $\Theta_{0}$ is truly set-valued. As was mentioned above, we
are not in a position to establish algorithmic convergence of our
Algorithm \ref{algo:gpd}. Nevertheless, assuming solutions found
by this (or any other) algorithm, say $\hat{\theta}_{n}$, are within
the $\tau_{n}$ tolerance from exact solutions in the sense that $\hat{\theta}_{n}\in\widehat{\Theta}_{n}(\tau_{n})$,
these consistency results would also apply to that algorithm. Therefore,
once algorithmic convergence is proven, our result gives a general
method of establishing consistency. 

Theorem \ref{thm:consistent} shows that estimation of the generative
representation parameters is consistent under appropriate assumptions.
Further properties of these estimators could also be entertained.
For instance, one could ask how to form confidence sets for these
solutions, or what is the asymptotic distribution of these solutions.
Such questions are studied in an IID minimax setting in \citet[Section 4]{Meitz:2024}
and \citet{Meitz:Shapiro:2025}, respectively. We leave the extension
of these results to the present setting for future research.

\subsection{Discussion: Weak convergence of the generative representation}

Combining the generative representation of Lemma \ref{lem:representation}
with the parameter estimators, say $\hat{\theta}_{n}=(\hat{\gamma}_{n},\hat{\delta}_{n})$,
now in hand, raises some obvious interesting questions: Does $G_{\hat{\gamma}_{n}}(Z_{t},X_{t})$
converge in distribution to $(Y_{t},X_{t})$, and does the conditional
distribution of $G_{\hat{\gamma}_{n}}(Z_{t},x)$ given $X_{t}=x$
converge to the conditional distribution of $Y_{t}$ given $X_{t}=x$?
\citet[Section 4]{Zhou:etal:2023} study such questions in their setup
and with IID data, establishing weak convergence. However, extending
their results to our framework and proving weak convergence for our
GPD method appears challenging for a number of technical reasons.
The primary obstacle is that key assumptions employed by \citet[p. 1841]{Zhou:etal:2023}
do not hold in our setup. In particular, for MLP generators with ReLU
activation, the push-forward distribution is generally not absolutely
continuous with respect to the Lebesgue measure, and therefore a density
is not guaranteed to exist (cf. the discussion around our Lemma \ref{lem:criterion}).
Extending their results to allow for more general distributions appears
highly non-trivial. Moreover, their approach requires the density
ratio between fake and real data distributions to be continuous, an
assumption that appears restrictive in our context. Yet another technical
hurdle is that their results rely on boundedness assumptions for both
the neural network architectures and the true generator $G^{*}$,
which implicitly restricts the support of the data distribution to
a compact set. Such constraints are incompatible with our setup, which
explicitly allows for unbounded covariate and outcome spaces under
Assumption \ref{asm:moments}, alongside temporal dependence ($\beta$-mixing).
Forcing these boundedness constraints would directly conflict with
the model's ability to process unbounded time series in our empirical
applications. Classic results by \citet{hornik1989multilayer} established
that networks are universal approximators, and numerous papers have
obtained explicit approximation error bounds. For instance, \citet{Barron:1993}
found rates for shallow networks --- later improved by \citet{Chen:White:1999}
for general sigmoid functions --- while \citet{Yarotsky2017} did
so for deep ReLU networks, and \citet{Zhou:etal:2023} rely on recent
results by \citet{Shen:Yang:Zhang:2020}. Crucially, all of the above
restrict the domain of the target functions to a compact set, such
as the unit hypercube. Without a compact domain, the uniform approximation
error between the true conditional distribution map and the neural
network can grow without bound in the tails. Finding approximation
error rates under both unbounded support and time series dependence
is still an open problem. Therefore, while we prove Hausdorff consistency
in Theorem \ref{thm:consistent}, we leave the derivation of weak
convergence for future research.

\section{Empirical applications\label{sec:empirical}}

We briefly outline the three applications that demonstrate the practical
benefits of our GPD approach. First, in a dynamic portfolio allocation
exercise (S\&P 500 returns), we show how directly sampling from the
predictive distribution effortlessly handles nonlinear utility functions
and data transformations without requiring closed-form analytical
solutions. Second, when forecasting S\&P 500 realized variance, we
demonstrate the GPD framework's ability to capture severe departures
from Gaussianity in levels, while maintaining performance comparable
to standard linear-Gaussian models when applied to logarithmically
transformed data. Finally, in modeling realized covariance matrices,
we highlight the framework's scalability to multivariate settings,
easily circumventing the analytical intractability of multivariate
matrix transformations.

\subsection{S\&P 500 returns}

We first illustrate the use of GPD in a dynamic portfolio allocation
exercise. This exercise may be viewed as a frequentist counterpart
to \citet{Kandel:Stambaugh:1996}. We consider monthly log-returns
on the S\&P 500 index from January 1927 to October 2025. Let $P_{t}$
denote the adjusted closing price at end of month $t$, and define
log-returns by $r_{t}=100(\ln P_{t}-\ln P_{t-1}).$ We also define
simple returns as $R_{t}=100(P_{t}/P_{t-1}-1)=100(\exp(r_{t}/100)-1).$
The distinction between $r_{t}$ and $R_{t}$ is economically relevant
(indeed, if $r_{t}$ were distributed standard normal, a mean-variance
investor (mistakenly) focusing on log-returns would conclude the optimal
allocation weight is 0, while the optimal mean-variance allocation
based on simple returns is close to one half). The sample consists
of 1,174 monthly observations. Descriptive statistics, reported in
Table \ref{tab:sp500return}(a), indicate pronounced non-Gaussian
features, including skewness and excess kurtosis, consistent with
well-known properties of equity returns.

\begin{figure}
\centering

\begin{minipage}[t]{0.55\columnwidth}%
\includegraphics[width=1\columnwidth]{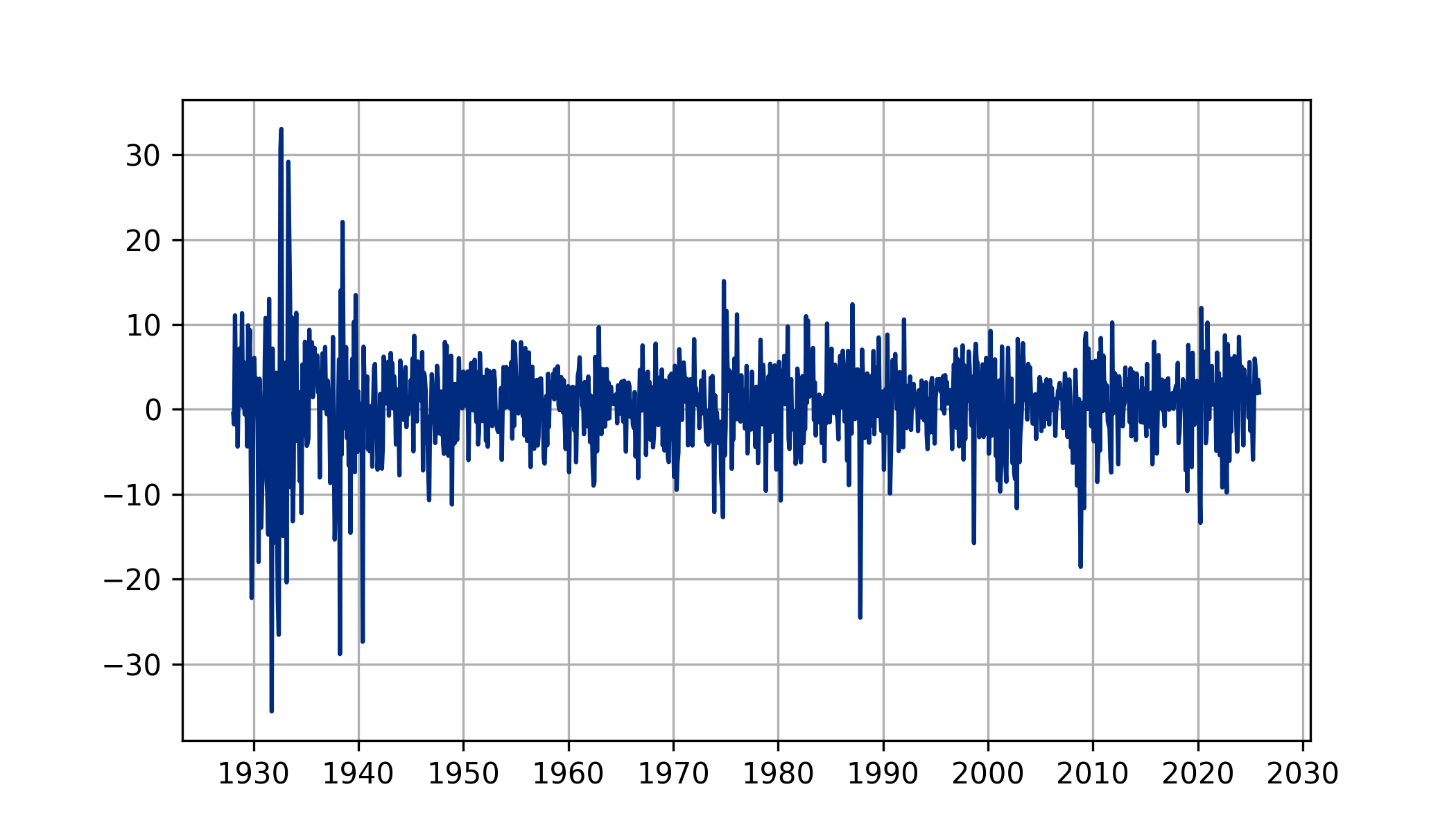}%
\end{minipage}%
\begin{minipage}[t]{0.45\columnwidth}%
\includegraphics[width=1\columnwidth]{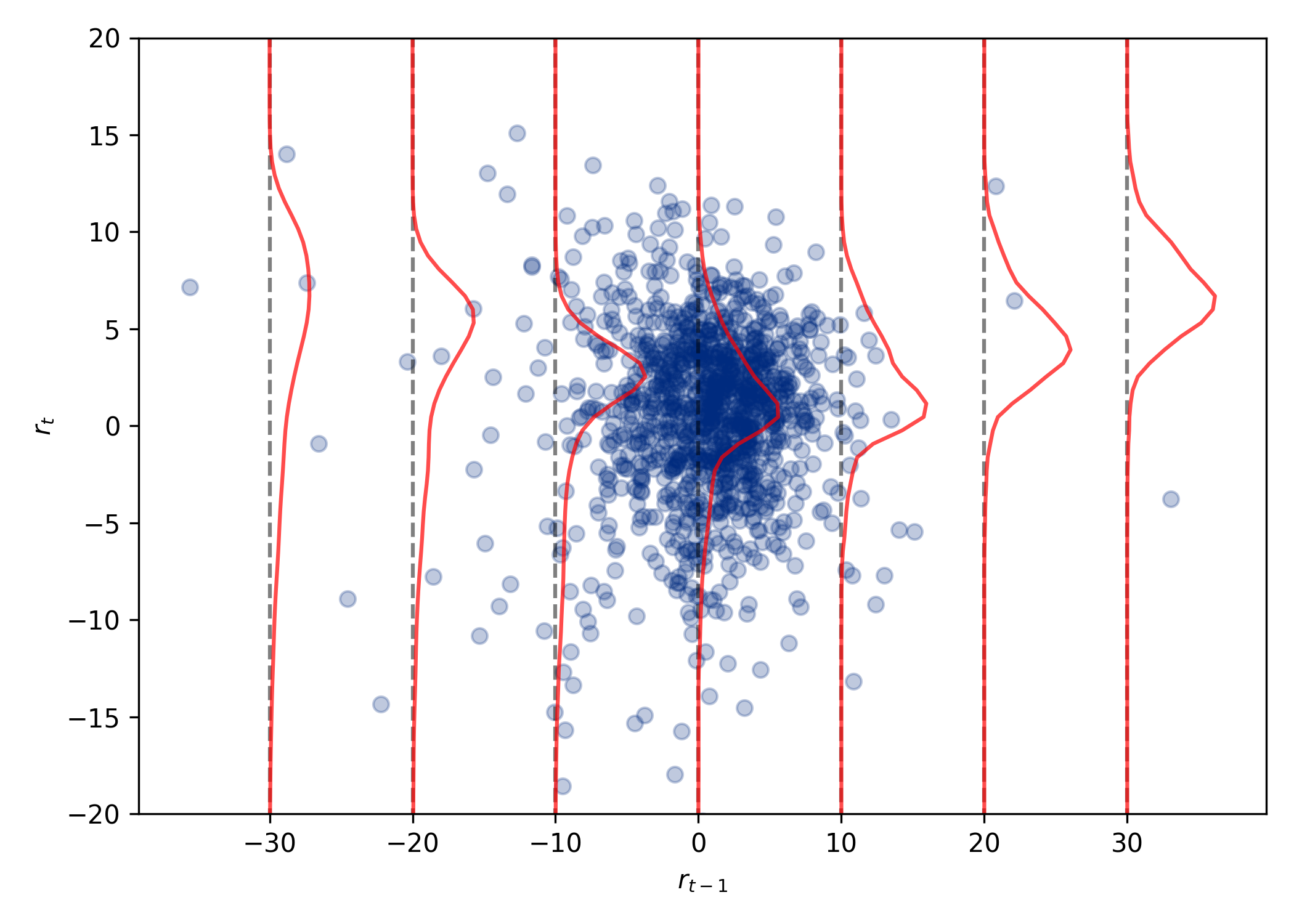}%
\end{minipage}

\begin{minipage}[t]{0.55\columnwidth}%
\includegraphics[width=1\columnwidth]{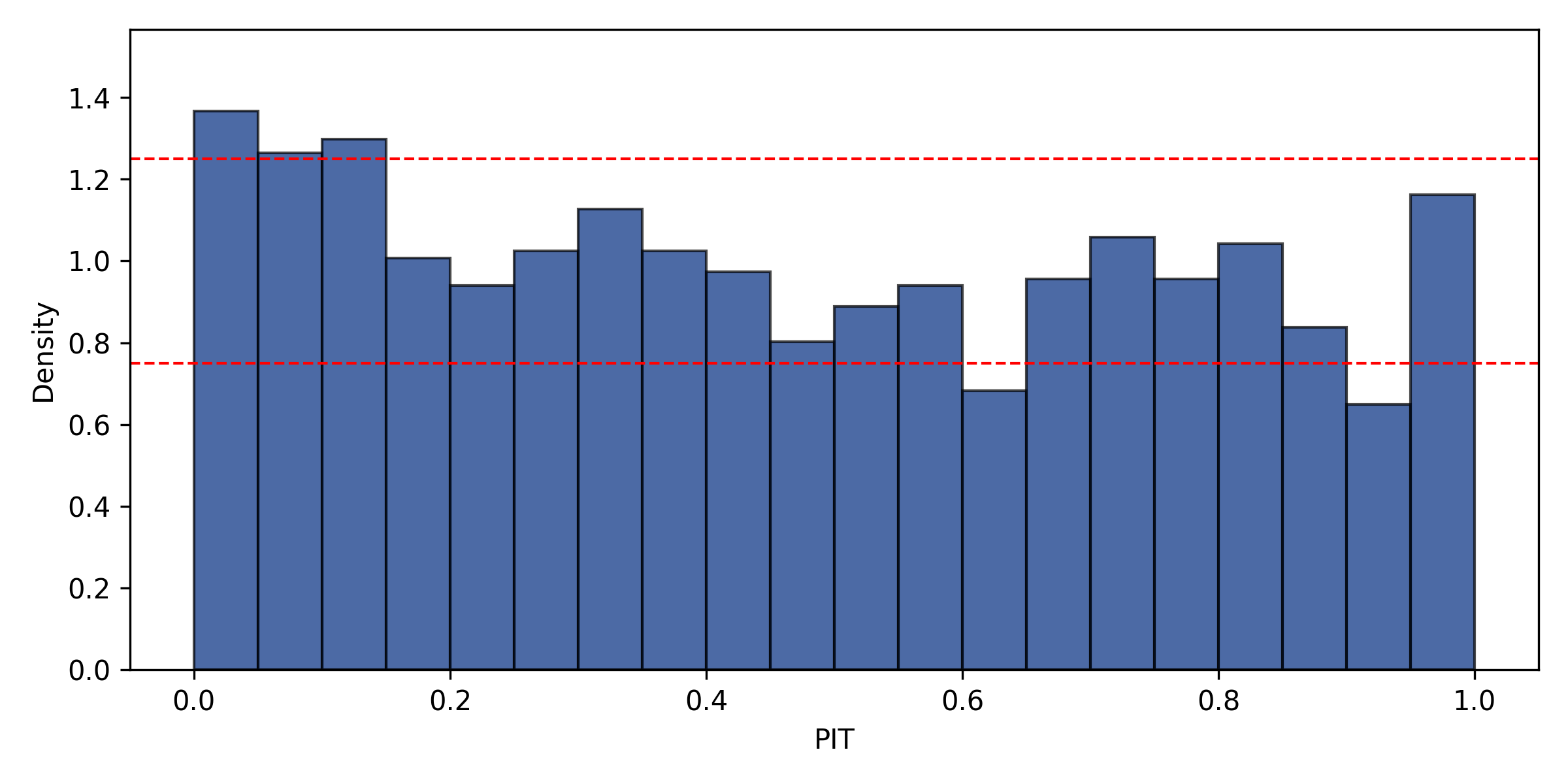}%
\end{minipage}%
\begin{minipage}[t]{0.45\columnwidth}%
\includegraphics[width=1\columnwidth]{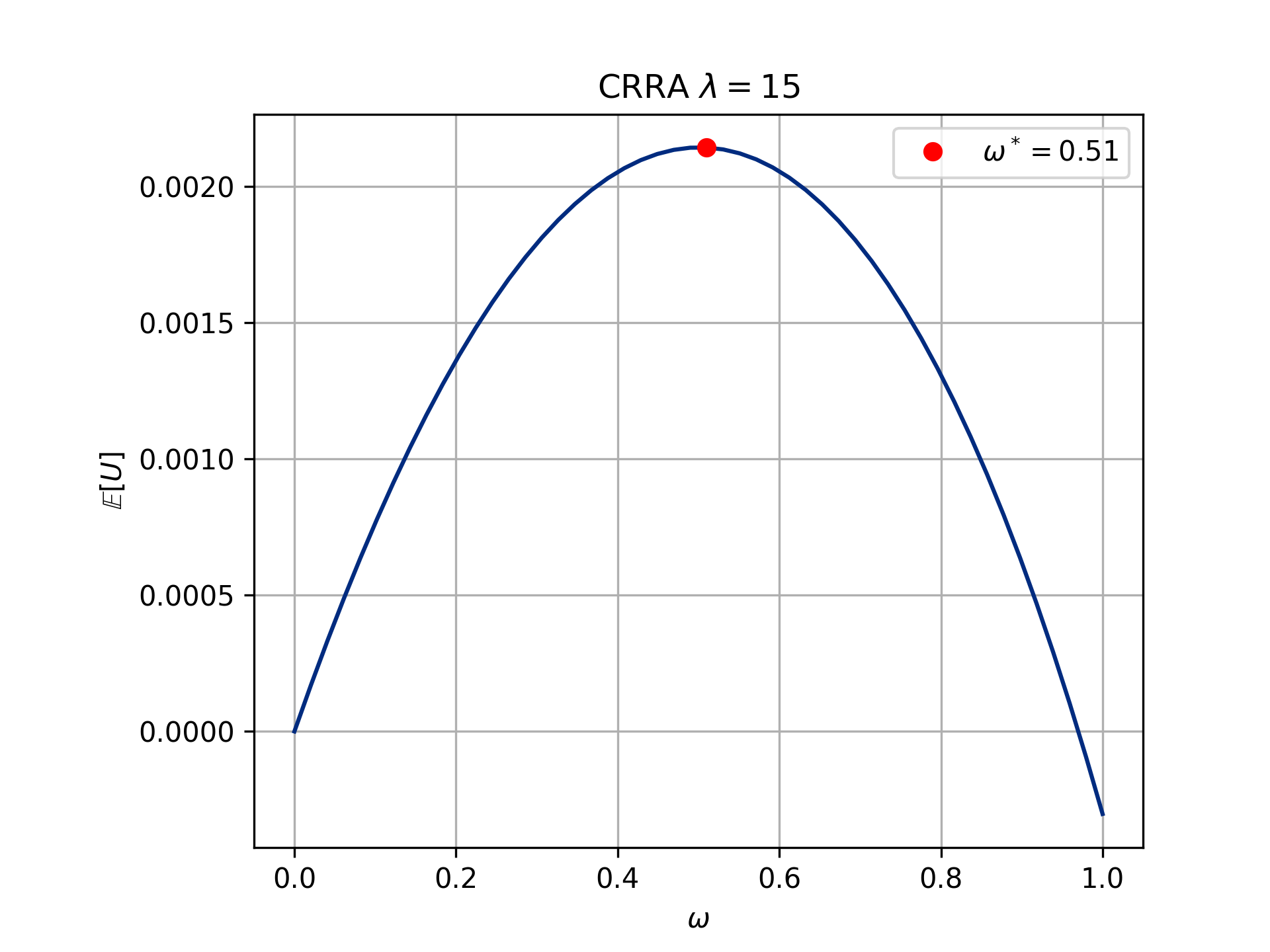}%
\end{minipage}

\caption{Top left: S\&P 500 log-returns, $r_{t}$. Top right: S\&P 500 log-returns
($r_{t}$) vs first lag (note: y-axis cropped to $[-20,20]$ to enhance
visibility). Kernel density summaries of the predictive distribution
based on 10,000 draws from GPD are depicted in red at selected values
of $r_{t-1}$, and with the remaining lags set to the sample mean
of $r_{t}$. Bottom left: In-sample PIT for GPD-based predictive distribution
and the remaining lags set to the sample mean. Bottom right: $\mathbb{E}\left[U\left(1+\omega R_{t}/100\right)\mid X_{t}\right]$,
with $r_{t}=100\ln\left(1+R_{t}/100\right)$ drawn from GPD-based
predictive distribution conditional on all lags set to the sample
mean.}

\label{fig:sp500return}
\end{figure}

\begin{table}
\centering

\subfloat[]{%
\begin{tabular}{lrr@{\extracolsep{0pt}.}lr@{\extracolsep{0pt}.}lr@{\extracolsep{0pt}.}lr@{\extracolsep{0pt}.}lr@{\extracolsep{0pt}.}lr@{\extracolsep{0pt}.}lr@{\extracolsep{0pt}.}l}
\noalign{\vskip\doublerulesep}
 & $n$ & \multicolumn{2}{c}{Mean} & \multicolumn{2}{c}{Std Dev} & \multicolumn{2}{c}{Skewness} & \multicolumn{2}{c}{Kurtosis} & \multicolumn{2}{c}{Minimum} & \multicolumn{2}{c}{Median} & \multicolumn{2}{c}{Maximum}\tabularnewline[\doublerulesep]
\hline 
\noalign{\vskip\doublerulesep}
\noalign{\vskip\doublerulesep}
$r_{t}$ & 1174 & 0&507 & 5&357 & -0&622 & 10&181 & -35&585 & 0&940 & 33&029\tabularnewline[\doublerulesep]
\noalign{\vskip\doublerulesep}
\noalign{\vskip\doublerulesep}
$R_{t}$ & 1174 & 0&652 & 5&337 & 0&105 & 10&811 & -29&942 & 0&945 & 39&138\tabularnewline[\doublerulesep]
\hline 
\noalign{\vskip\doublerulesep}
\end{tabular}}\ \subfloat[]{%
\begin{tabular}{lr@{\extracolsep{0pt}.}lr@{\extracolsep{0pt}.}lr@{\extracolsep{0pt}.}lr@{\extracolsep{0pt}.}lr@{\extracolsep{0pt}.}lr@{\extracolsep{0pt}.}lr@{\extracolsep{0pt}.}lr@{\extracolsep{0pt}.}l}
\noalign{\vskip\doublerulesep}
 & \multicolumn{2}{c}{$\lambda=5$} & \multicolumn{2}{c}{$\lambda=15$} & \multicolumn{2}{c}{$\lambda=50$} & \multicolumn{2}{c}{} & \multicolumn{2}{c}{Mean} & \multicolumn{2}{c}{Std Dev} & \multicolumn{2}{c}{Skewness} & \multicolumn{2}{c}{Kurtosis}\tabularnewline[\doublerulesep]
\hline 
\noalign{\vskip\doublerulesep}
\noalign{\vskip\doublerulesep}
$r_{t-1}=-10$ & 0&531 & 0&184 & 0&061 & \multicolumn{2}{c}{} & 0&468 & 4&469 & -1&720 & 5&896\tabularnewline[\doublerulesep]
\noalign{\vskip\doublerulesep}
\noalign{\vskip\doublerulesep}
$r_{t-1}=-5$ & 0&000 & 0&000 & 0&000 & \multicolumn{2}{c}{} & -0&235 & 3&663 & -1&546 & 5&895\tabularnewline[\doublerulesep]
\noalign{\vskip\doublerulesep}
\noalign{\vskip\doublerulesep}
$r_{t-1}=0$ & 1&000 & 0&592 & 0&184 & \multicolumn{2}{c}{} & 1&033 & 3&357 & -0&664 & 4&835\tabularnewline[\doublerulesep]
\noalign{\vskip\doublerulesep}
\noalign{\vskip\doublerulesep}
$r_{t-1}=5$ & 1&000 & 0&612 & 0&184 & \multicolumn{2}{c}{} & 0&713 & 2&699 & -0&777 & 5&403\tabularnewline[\doublerulesep]
\noalign{\vskip\doublerulesep}
\noalign{\vskip\doublerulesep}
$r_{t-1}=10$ & 1&000 & 1&000 & 0&469 & \multicolumn{2}{c}{} & 2&030 & 2&745 & -0&161 & 4&265\tabularnewline[\doublerulesep]
\hline 
\noalign{\vskip\doublerulesep}
\end{tabular}}

\caption{(a) Descriptive statistics for S\&P 500 monthly returns. (b) Optimal
portfolio allocation $\omega_{t}^{*}$ based on GPD, and GPD-based
conditional moments for different values of lagged returns. Remaining
lags set to the sample mean.}

\label{tab:sp500return}
\end{table}

To use our GPD method, we consider $Y_{t}=r_{t}$, choose the conditioning
variables as $X_{t}=(r_{t-1},r_{t-2},r_{t-3})$, and set $Z_{t}\overset{\text{IID}}{\sim}\mathcal{N}\left(0,1\right)$.
The generator and discriminator are both ReLU-MLPs with two hidden
layers and 50 hidden units per layer. No other architecture choices
are required from the user. Estimation is then performed using Algorithm
\ref{algo:gpd}, with the following standard hyperparameter choices
reported here for completeness: the learning rate in both Adam optimizers
is set to $2\cdot10^{-4}$, the batch size is 256, and all other hyperparameters
are set to their default values from the PyTorch implementation. We
implemented the early stopping rule, detailed in Appendix \ref{subsec:swd},
whose patience parameter is set to 750.

The estimated generator provides a predictive distribution for log-returns
$r_{t}$ given past returns contained in $X_{t}$. Because the generator
delivers an explicit mapping from Gaussian innovations to future returns
conditional on $X_{t}$, it is straightforward to simulate from the
predictive distribution at any conditioning value. Kernel density
summaries of the predictive distribution reveal nonlinearities and
conditional heteroskedasticity, see Figure \ref{fig:sp500return}
(top right). In particular, the shape of the predictive distribution
varies with lagged returns, exhibiting skewness and changes in dispersion.
As shown in Figure \ref{fig:sp500return} (bottom left), goodness
of fit diagnostics based on the Probability Integral Transform (PIT)
indicate an adequate fit \citep{Diebold:Gunther:Tay:1998}. Specifically,
if the conditional distribution is correctly specified, the PIT series
should be uniformly distributed over the unit interval. Importantly,
the predictive distribution for simple returns $R_{t}$ is obtained
directly by applying the nonlinear transformation $R_{t}=100(\exp(r_{t}/100)-1)$
to simulated draws from GPD. No additional approximation or analytical
argument is required.

We now consider a risk-averse investor with constant relative risk
aversion utility (CRRA). The utility as a function of wealth, $W$,
is given by
\[
U\left(W\right)=\frac{W^{1-\lambda}-1}{1-\lambda},\quad\lambda>1.
\]
The investor allocates a fraction $\omega_{t}\in\left[0,1\right]$
of wealth to the risky asset. Specifically, they choose
\[
\omega_{t}^{*}=\underset{\omega\in\left[0,1\right]}{\arg\max}\:\mathbb{E}\left[U\left(1+\omega\frac{R_{t}}{100}\right)\mid X_{t}\right].
\]
In general, this optimization problem has no closed-form solution.
However, given the GPD-based predictive distribution, the objective
function can be evaluated by Monte Carlo integration using simulated
draws from the generator. This yields an approximation to expected
utility as a function of $\omega$, from which the optimal allocation
can be computed numerically (see Figure \ref{fig:sp500return}, bottom
right).

The GPD-based predictive distribution allows us to obtain optimal
portfolio weights across different conditioning scenarios, see Table
\ref{tab:sp500return}(b) for some examples. For all levels of risk
aversion considered, the optimal allocation responds nonlinearly to
lagged returns. For example, when recent returns are negative ($r_{t-1}=-5)$,
the optimal allocation collapses to zero. Conversely, after sufficiently
positive lagged returns, the optimal weight may increase sharply,
sometimes reaching the upper constraint. The sensitivity to the coefficient
of relative risk aversion is also substantial. As expected, optimal
weights shrink to zero as the coefficient of risk aversion ($\lambda$)
increases. These allocations arise partly from the shape of the predictive
distribution and would be missed by approaches that focus only on
conditional means and variances. Again, we emphasize that the optimal
allocation should be based on simple rather than log-returns, which
involves the predictive distribution of $R_{t}$, a nonlinear transformation
of $r_{t}$. Note that while all GPD-based conditional moments take
a role in determining the exact allocations, the first two are already
quite helpful to qualitatively explain all optimal allocation decisions,
as it is evidenced from Table \ref{tab:sp500return}.

We emphasize that expected utility depends solely on the first two
moments of returns if either investors have quadratic utility or if
simple returns are Gaussian and preferences exhibit constant absolute
risk aversion (CARA). Outside these cases, higher-order features of
the predictive distribution such as skewness or kurtosis directly
affect optimal decisions. This is particularly relevant in this application,
where returns are clearly non-Gaussian and preferences are of the
CRRA type. Moreover, expected utility is defined over simple returns
rather than log-returns, introducing an additional nonlinear transformation.
In such settings, portfolio choice depends on the full predictive
distribution of returns, not just its mean and variance. This application
highlights a key advantage of the proposed framework: once the predictive
distribution is represented in terms of an efficient sampler, decision
problems relying on the predictive distribution can be solved straightforwardly.
Without direct access to the predictive distribution via a generator,
an investor would need to rely on cumbersome and computationally intensive
techniques. For instance, evaluating expected utility over non-Gaussian
returns typically requires high-dimensional numerical integration,
Taylor series approximations of the utility function \citep{Guidolin:Timmermann:2008},
or Markov Chain Monte Carlo (MCMC) simulations if a Bayesian parametric
model is assumed (e.g., \citealp{Barberis:2000}).

\subsection{S\&P 500 realized variance}

We next turn to forecasting realized variance of the S\&P 500, computed
from 5-minute intraday returns following \citet{Bollerslev:Patton:Quaedvlieg:2016},
for the period ranging from 21-Apr-1997 to 30-Aug-2013. The unconditional
distribution exhibits a heavy right tail and right skewness (see Table
\ref{tab:sp500rv_stats} and Figure \ref{fig:RV}). A logarithmic
transformation has the effect of eliminating the heavy tail and reducing
skewness. The transformation also offers the advantage that the support
of the distribution is unrestricted.

The empirical exercise is conducted in a static window framework.
All competing models are estimated on the first 2,000 observations
and evaluated out of sample on the remaining 2,096 observations. Forecast
performance is assessed using three standard loss functions: (Gaussian)
quasi-likelihood (QLIKE), mean squared error (MSE), and mean absolute
error (MAE). Both QLIKE and MSE are of particular interest in this
context, as they are robust to measurement error in realized variance
and penalize forecast errors in ways that are directly relevant for
volatility forecasting \citep{Patton:2011}. Among the two, QLIKE
is often preferred in empirical comparisons because it typically provides
greater power in forecast evaluation tests and is less dominated by
extreme realizations, while still remaining sensitive to systematic
differences in predictive accuracy.

\begin{figure}[p]
\centering%
\begin{minipage}[t]{0.9\columnwidth}%
\includegraphics[width=1\columnwidth]{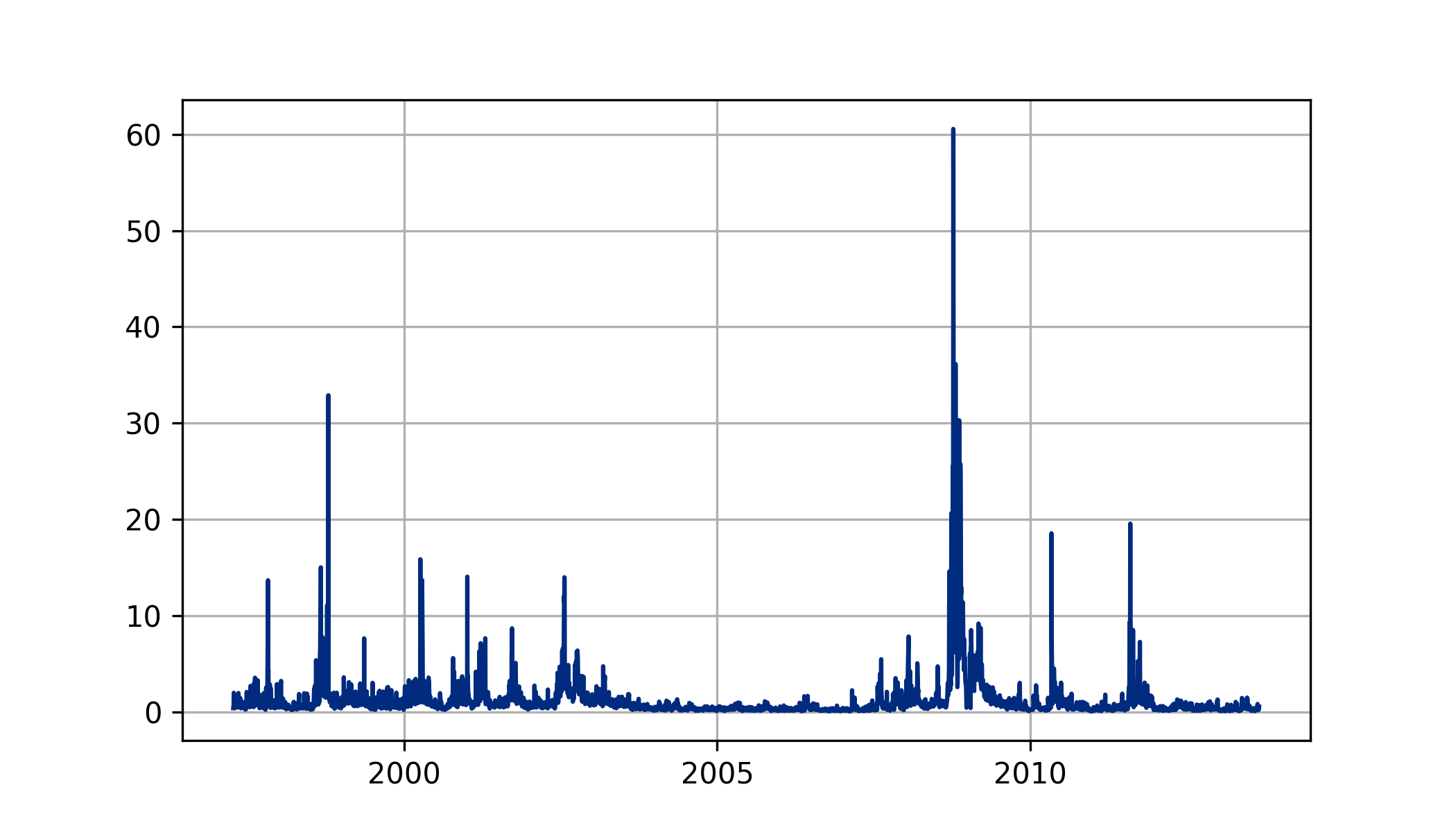}%
\end{minipage}

\caption{Five minute RV for S\&P 500.}
\label{fig:RV}
\end{figure}

\begin{table}[p]
\centering%
\begin{tabular}{lr@{\extracolsep{0pt}.}lr@{\extracolsep{0pt}.}lr@{\extracolsep{0pt}.}lr@{\extracolsep{0pt}.}lr@{\extracolsep{0pt}.}lr@{\extracolsep{0pt}.}lr@{\extracolsep{0pt}.}lr@{\extracolsep{0pt}.}l}
\noalign{\vskip\doublerulesep}
 & \multicolumn{2}{c}{$n$} & \multicolumn{2}{c}{Mean} & \multicolumn{2}{c}{Std Dev} & \multicolumn{2}{c}{Skewness} & \multicolumn{2}{c}{Kurtosis} & \multicolumn{2}{c}{Minimum} & \multicolumn{2}{c}{Median} & \multicolumn{2}{c}{Maximum}\tabularnewline[\doublerulesep]
\hline 
\noalign{\vskip\doublerulesep}
\noalign{\vskip\doublerulesep}
$RV_{t}$ & \multicolumn{2}{c}{4096} & 1&175 & 2&315 & 10&028 & 166&922 & 0&043 & 0&629 & 60&563\tabularnewline[\doublerulesep]
\noalign{\vskip\doublerulesep}
\noalign{\vskip\doublerulesep}
$\ln RV_{t}$ & \multicolumn{2}{c}{4096} & -0&417 & 0&974 & 0&517 & 3&550 & -3&140 & -0&463 & 4&104\tabularnewline[\doublerulesep]
\hline 
\noalign{\vskip\doublerulesep}
\end{tabular}

\caption{Descriptive statistics for S\&P 500 realized variance.}
\label{tab:sp500rv_stats}
\end{table}

\begin{table}[p]
\centering\subfloat[All models estimated on the first 2,000 observations of $RV_{t}$
and evaluated out of sample on the remaining 2,096. GPD(5) with $L_{G}=2$,
$H_{G}=32$, $L_{D}=2$, $H_{D}=150$, and $X_{t}=(RV_{t-1},\dots,RV_{t-5})$.]{\centering

\begin{tabular}{lccc}
\noalign{\vskip\doublerulesep}
 & QLIKE & MSE & MAE\tabularnewline[\doublerulesep]
\hline 
\noalign{\vskip\doublerulesep}
\noalign{\vskip\doublerulesep}
AR(5) & 0.211 & 3.618{*} & 0.564\tabularnewline[\doublerulesep]
\noalign{\vskip\doublerulesep}
\noalign{\vskip\doublerulesep}
HAR & 0.181 & \textbf{3.573{*}} & 0.534\tabularnewline[\doublerulesep]
\noalign{\vskip\doublerulesep}
\noalign{\vskip\doublerulesep}
MSAR(5) & 0.204 & 3.919{*} & 0.571\tabularnewline[\doublerulesep]
\noalign{\vskip\doublerulesep}
\noalign{\vskip\doublerulesep}
GPD(5) & \textbf{0.147{*}} & 3.712{*} & \textbf{0.521{*}}\tabularnewline[\doublerulesep]
\noalign{\vskip\doublerulesep}
\noalign{\vskip\doublerulesep}
RW & 0.19 & 5.146{*} & 0.575{*}\tabularnewline[\doublerulesep]
\noalign{\vskip\doublerulesep}
\noalign{\vskip\doublerulesep}
Mean & 0.782 & 8.201{*} & 1.187\tabularnewline[\doublerulesep]
\hline 
\noalign{\vskip\doublerulesep}
\end{tabular}

}\ \ \ \subfloat[All models estimated on the first 2,000 observations of $\ln RV_{t}$
and evaluated (on $RV_{t}$) out of sample on the remaining 2,096.
GPD$_{\log}$(5) with $L_{G}=2$, $H_{G}=32$, $L_{D}=2$, $H_{D}=150$,
and $X_{t}=(\ln RV_{t-1},\dots,\ln RV_{t-5})$.]{\centering

\begin{tabular}{lccc}
\noalign{\vskip\doublerulesep}
 & QLIKE & MSE & MAE\tabularnewline[\doublerulesep]
\hline 
\noalign{\vskip\doublerulesep}
\noalign{\vskip\doublerulesep}
AR$_{\log}$(5) & 0.142{*} & 3.458{*} & 0.493{*}\tabularnewline[\doublerulesep]
\noalign{\vskip\doublerulesep}
\noalign{\vskip\doublerulesep}
ARMA$_{\log}$(1,1) & 0.144{*} & 3.516{*} & 0.496{*}\tabularnewline[\doublerulesep]
\noalign{\vskip\doublerulesep}
\noalign{\vskip\doublerulesep}
HAR$_{\log}$ & 0.144{*} & 3.515{*} & \textbf{0.489}{*}\tabularnewline[\doublerulesep]
\noalign{\vskip\doublerulesep}
\noalign{\vskip\doublerulesep}
MSAR$_{\log}$(5) & \textbf{0.140}{*} & 3.497{*} & 0.497{*}\tabularnewline[\doublerulesep]
\noalign{\vskip\doublerulesep}
\noalign{\vskip\doublerulesep}
GPD$_{\log}$(5) & 0.143{*} & \textbf{3.200{*}} & 0.501{*}\tabularnewline[\doublerulesep]
\noalign{\vskip\doublerulesep}
\noalign{\vskip\doublerulesep}
Mean$_{\log}$ & 0.782 & 8.200{*} & 1.157\tabularnewline[\doublerulesep]
\hline 
\noalign{\vskip\doublerulesep}
\end{tabular}

}

\caption{Forecast comparison S\&P 500, 5-min RV. The superscript $^{*}$ denotes
inclusion in the model confidence set at the 90\% level.}

\label{tab:sp500rv}
\end{table}

Inference is conducted using a Model Confidence Set (MCS) \citep{Hansen:Lunde:Nason:2011}
procedure that is implemented jointly across all competing models,
including specifications estimated directly on realized variance and
those estimated on log realized variance. The first comparison, reported
in Table \ref{tab:sp500rv}(a), considers models estimated directly
on realized variance levels. For the GPD, we set $Y_{t}=RV_{t}$,
$X_{t}=(RV_{t-1},\dots,RV_{t-5})$, $Z_{t}\sim\mathcal{N}(0,1)$,
and network architecture with $L_{G}=L_{D}=2$, $H_{G}=32$, and $H_{D}=150$.
Among traditional benchmarks such as AR(5), Heterogeneous Autoregressive
model (HAR), Markov-Switching Autoregressive Model of order 5 (MSAR(5)),
random walk (RW), and the unconditional mean, the GPD(5) specification
delivers the lowest QLIKE loss and is included in the 90\% model confidence
set. The gains relative to traditional models are particularly evident
under QLIKE and MAE.

Table \ref{tab:sp500rv}(b) reports results for models estimated on
log RV and evaluated on the realized variance scale. For the GPD,
$Y_{t}=\ln RV_{t}$ and $X_{t}=(\ln RV_{t-1},\dots,\ln RV_{t-5})$
and other choices are as before, with the only exception being the
patience parameter in the early stopping rule, which is set to 40
for the logarithmic scale and to 500 for the realized variance scale.
As expected, the log transformation creates a more favorable environment
for conventional linear and nonlinear time series models. In this
setting, AR, ARMA, HAR and MSAR specifications all perform well and
belong to the 90\% model confidence set. The GPD specification remains
competitive, particularly under MSE, although its relative advantage
under QLIKE is attenuated once the data are transformed.

Taken together, these results highlight two key insights. First, when
forecasting realized variance directly in levels, where nonlinearity
and tail risks are most pronounced, the GPD delivers clear gains in
QLIKE and MAE. Second, when transformations are applied that bring
the data distribution closer to Gaussianity, standard models perform
well, and GPD still remains competitive while retaining the ability
to generate full predictive distributions and to handle nonlinear
transformations without analytical approximations. This flexibility
is valuable in forecasting volatility and can be further exploited
where interest extends beyond point forecasts to the entire distribution
of RV.

\subsection{Realized covariance}

Finally, we demonstrate the scalability of GPD to multivariate time
series by modeling and forecasting realized covariance matrices. The
data consist of open-to-close daily realized covariance matrices for
two large US financial institutions, Bank of America (BAC) and JPMorgan
Chase (JPM), computed using 5-minute intraday returns over the period
from 1-Feb-2001 to 31-Dec-2009. The construction of the realized covariance
matrices follows the methodology of \citet{Noureldin:Shephard:Sheppard:2012}.
Let $RV_{t}$ denote the 2$\times$2 realized covariance matrix on
day $t$.

To ensure positive definiteness, we consider a transformation based
on the matrix logarithm. For a symmetric positive definite matrix
$RV_{t}$, the matrix logarithm is defined via its eigendecomposition
\[
RV_{t}=Q_{t}\left[\begin{array}{cc}
\lambda_{1,t} & 0\\
0 & \lambda_{2,t}
\end{array}\right]Q_{t}',\qquad\ln RV_{t}:=Q_{t}\left[\begin{array}{cc}
\ln\lambda_{1,t} & 0\\
0 & \ln\lambda_{2,t}
\end{array}\right]Q_{t}'.
\]
The transformed series is again vectorized, $Y_{t}=\text{vech}\left(\ln RV_{t}\right)$,
where $\text{vech}\left(\cdot\right)$ is the half-vectorization operator.
Working with the matrix logarithm has the advantage of producing series
with distributions that are closer to Gaussian, see Table \ref{tab:rcov_stats}.

\begin{table}
\centering%
\begin{tabular}{lr@{\extracolsep{0pt}.}lr@{\extracolsep{0pt}.}lr@{\extracolsep{0pt}.}lr@{\extracolsep{0pt}.}lr@{\extracolsep{0pt}.}lr@{\extracolsep{0pt}.}lr@{\extracolsep{0pt}.}lr@{\extracolsep{0pt}.}lr@{\extracolsep{0pt}.}lr@{\extracolsep{0pt}.}lr@{\extracolsep{0pt}.}l}
\noalign{\vskip\doublerulesep}
 & \multicolumn{6}{c}{$RV_{t}$} & \multicolumn{2}{c}{} & \multicolumn{6}{c}{$\ln RV_{t}$} & \multicolumn{2}{c}{} & \multicolumn{6}{c}{$\text{GPD}:\ln RV_{t}$}\tabularnewline[\doublerulesep]
\cline{2-7}\cline{10-15}\cline{18-23}
\noalign{\vskip\doublerulesep}
\noalign{\vskip\doublerulesep}
 & \multicolumn{2}{c}{BAC} & \multicolumn{2}{c}{Cov} & \multicolumn{2}{c}{JPM} & \multicolumn{2}{c}{} & \multicolumn{2}{c}{BAC} & \multicolumn{2}{c}{Cov} & \multicolumn{2}{c}{JPM} & \multicolumn{2}{c}{} & \multicolumn{2}{c}{BAC} & \multicolumn{2}{c}{Cov} & \multicolumn{2}{c}{JPM}\tabularnewline[\doublerulesep]
\cline{1-7}\cline{10-15}\cline{18-23}
\noalign{\vskip\doublerulesep}
\noalign{\vskip\doublerulesep}
$n$ & \multicolumn{2}{c}{2242} & \multicolumn{2}{c}{2242} & \multicolumn{2}{c}{2242} & \multicolumn{2}{c}{} & \multicolumn{2}{c}{2242} & \multicolumn{2}{c}{2242} & \multicolumn{2}{c}{2242} & \multicolumn{2}{c}{} & \multicolumn{2}{c}{2242} & \multicolumn{2}{c}{2242} & \multicolumn{2}{c}{2242}\tabularnewline[\doublerulesep]
\noalign{\vskip\doublerulesep}
mean & 5&468 & 3&048 & 5&055 & \multicolumn{2}{c}{} & 0&130 & 0&565 & 0&552 & \multicolumn{2}{c}{} & 0&115 & 0&574 & 0&559\tabularnewline[\doublerulesep]
\noalign{\vskip\doublerulesep}
std & 16&811 & 8&295 & 11&094 & \multicolumn{2}{c}{} & 1&353 & 0&269 & 1&192 & \multicolumn{2}{c}{} & 1&308 & 0&267 & 1&166\tabularnewline[\doublerulesep]
\noalign{\vskip\doublerulesep}
skew & 7&173 & 6&515 & 7&524 & \multicolumn{2}{c}{} & 1&015 & 0&592 & 0&461 & \multicolumn{2}{c}{} & 0&953 & 0&451 & 0&380\tabularnewline[\doublerulesep]
\noalign{\vskip\doublerulesep}
kurtosis & 72&505 & 59&722 & 84&610 & \multicolumn{2}{c}{} & 3&871 & 3&368 & 2&811 & \multicolumn{2}{c}{} & 3&737 & 3&429 & 2&521\tabularnewline[\doublerulesep]
\noalign{\vskip\doublerulesep}
min & 0&074 & -0&432 & 0&112 & \multicolumn{2}{c}{} & -2&664 & -0&134 & -2&228 & \multicolumn{2}{c}{} & -2&293 & -0&292 & -2&093\tabularnewline[\doublerulesep]
\noalign{\vskip\doublerulesep}
median & 1&046 & 0&643 & 1&867 & \multicolumn{2}{c}{} & -0&165 & 0&528 & 0&459 & \multicolumn{2}{c}{} & -0&131 & 0&554 & 0&490\tabularnewline[\doublerulesep]
\noalign{\vskip\doublerulesep}
max & 277&308 & 112&197 & 176&478 & \multicolumn{2}{c}{} & 5&480 & 1&714 & 5&058 & \multicolumn{2}{c}{} & 5&413 & 1&776 & 4&472\tabularnewline[\doublerulesep]
\hline 
\end{tabular}\caption{Descriptive statistics of realized covariance of BAC and JPM and its
matrix logarithm. The statistics for $\text{GPD}:\ln RV_{t}$ are
based on the unconditional distribution implied by the generator.
The latter is constructed by aggregating simulated draws from the
generator, $G_{\hat{\gamma}_{n}}(Z_{t},X_{t})$, evaluated over the
entire empirical sequence of historical predictors $X_{t}$ using
independent standard normal innovations $Z_{t}$ of dimension 3.}

\label{tab:rcov_stats}
\end{table}

Moreover, while forecasts of $RV_{t}$ are restricted to lie in the
cone of positive definite matrices, the components of the matrix logarithm
are unrestricted. It is possible to estimate GPD directly on the raw
realized covariances $RV_{t}$; however, the positive definiteness
of the outputs cannot be guaranteed without modifications. These modifications
may involve restricting the parameter space $\Gamma$ (with corresponding
changes to the training procedure) or designing a suitably constrained
generator.

We remark that, whereas the realized covariance literature has primarily
focused on forecasting the implied covariance matrix --- namely,
the conditional mean of $RV_{t}$ under conditional unbiasedness ---
forecasting $RV_{t}$ from a model specified for $\ln RV_{t}$ requires
the full predictive distribution, even under the assumption of Gaussianity.
Unlike in the univariate case, where under Gaussianity, realized variance
follows a log-normal distribution and a closed form expression for
the mean is known, in the multivariate case there is no closed form
expression for the mean of the matrix exponential of a multivariate
Gaussian.

Here, GPD is specified with a multivariate response and lagged predictors,
in analogy with a vector autoregressive structure of order 1. The
response variable is $Y_{t}=\text{vech}\left(\ln RV_{t}\right)$,
depending on the specification. The predictors consist of the first
lag of the response, i.e. $X_{t}=Y_{t-1}$, and latent innovations
$Z_{t}$ are drawn independently from a standard normal distribution
of dimension 3. The generator is implemented as a ReLU-MLP with two
hidden layers ($L_{G}=2$) of 16 hidden units ($H_{G}=16$), while
the discriminator is a ReLU-MLP with two hidden layers of 150 units
each ($L_{D}=2,H_{D}=150)$. Again, estimation is performed using
Algorithm \ref{algo:gpd}, with all other choices as above except
for the patience parameter in the early stopping rule, which in this
case is set to 500.

As seen in Figure \ref{fig:rcov} (top left), the unconditional multivariate
distribution of realized covariances exhibits pronounced right skewness,
heavy tails, and strong nonlinear dependence between variances and
covariances. In contrast, the unconditional marginal distributions
of the matrix logarithm of realized covariances are substantially
more symmetric, with markedly reduced skewness and kurtosis (see Figure
\ref{fig:rcov}, top right), although some non-Gaussian features are
still apparent. Moreover, there is strong (and somewhat nonlinear)
dependence between the components of both the raw realized covariance
and its matrix logarithm.

\begin{figure}
\centering

\begin{minipage}[t]{0.45\columnwidth}%
\includegraphics[width=1\columnwidth]{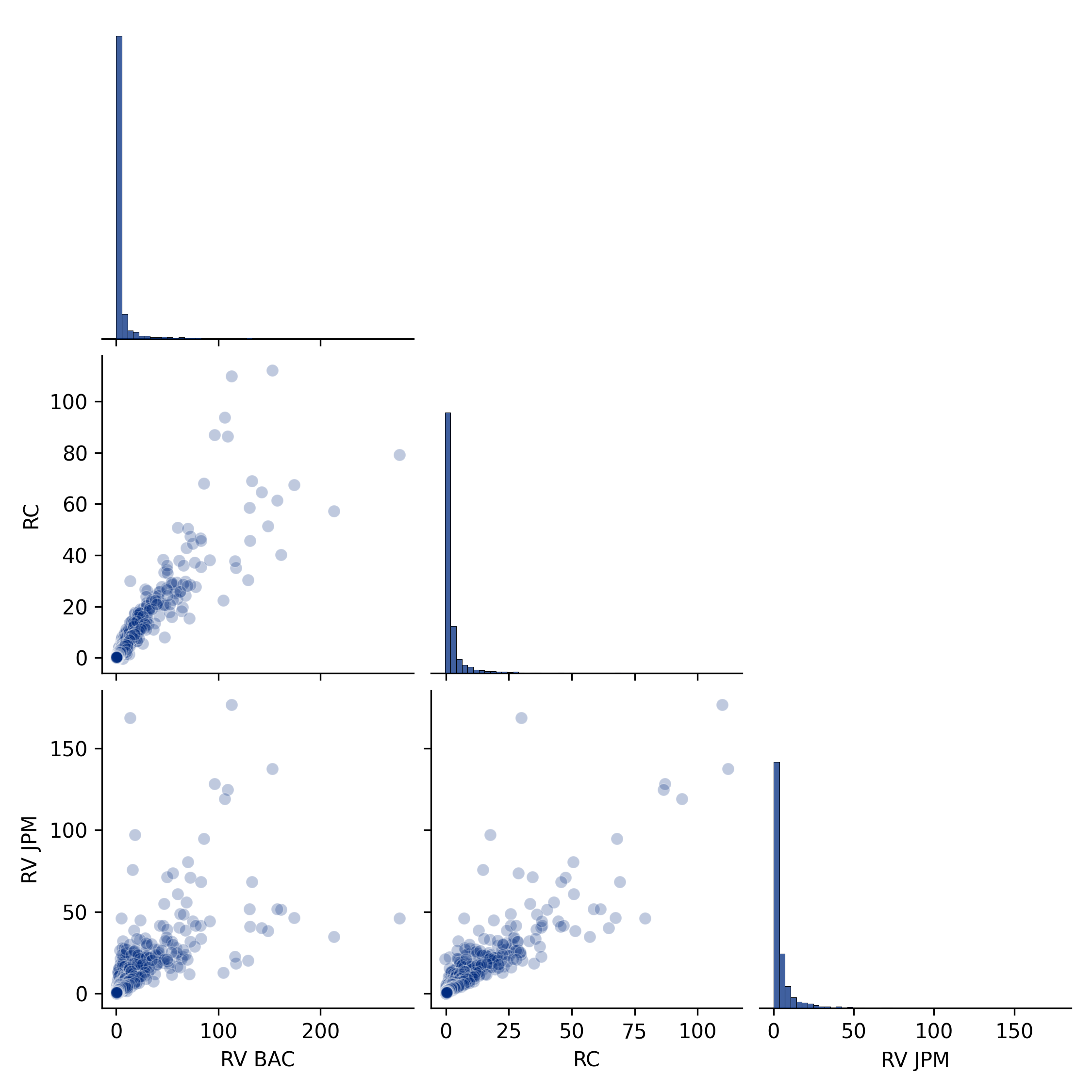}%
\end{minipage}%
\begin{minipage}[t]{0.45\columnwidth}%
\includegraphics[width=1\columnwidth]{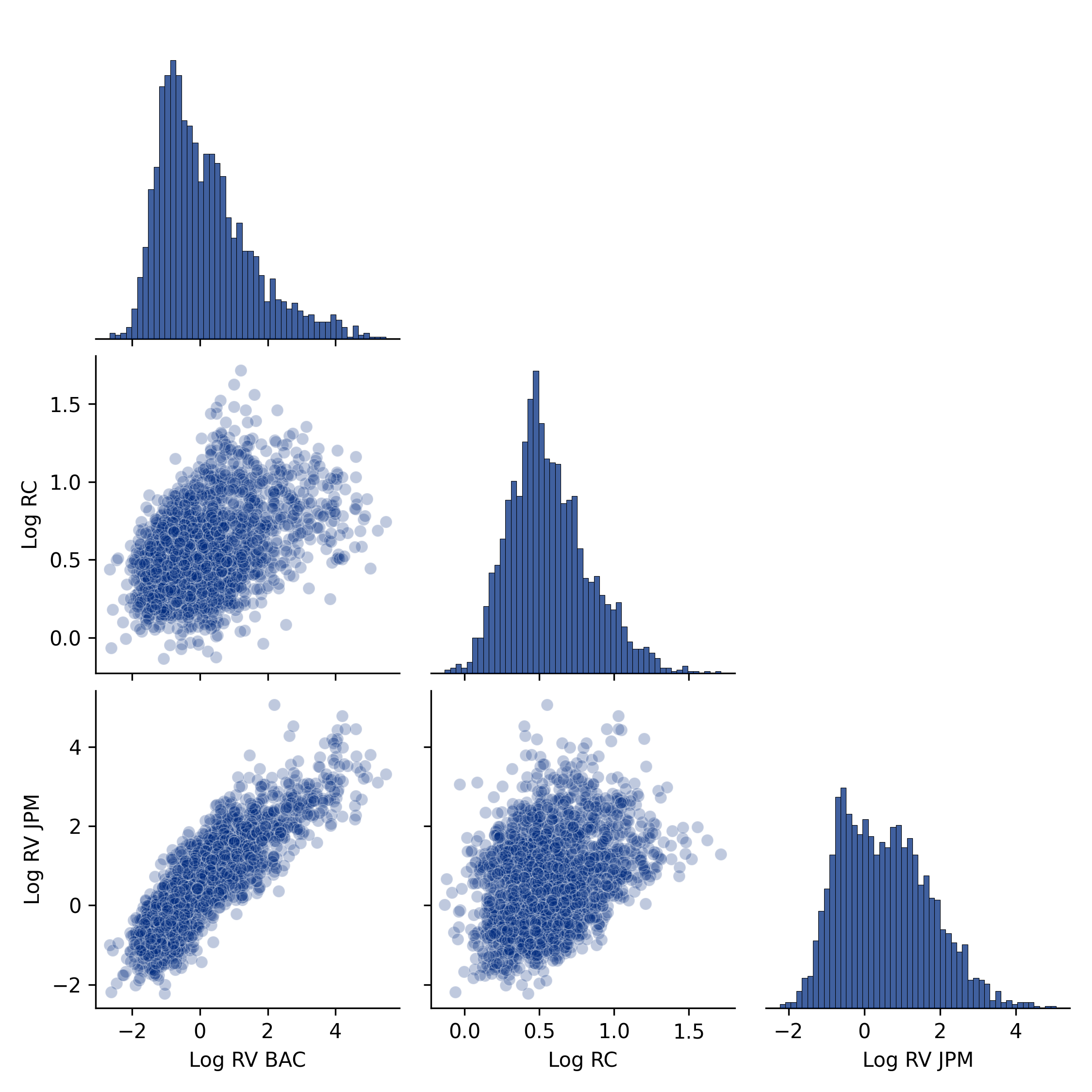}%
\end{minipage}

\begin{minipage}[t]{0.45\columnwidth}%
\includegraphics[width=1\columnwidth]{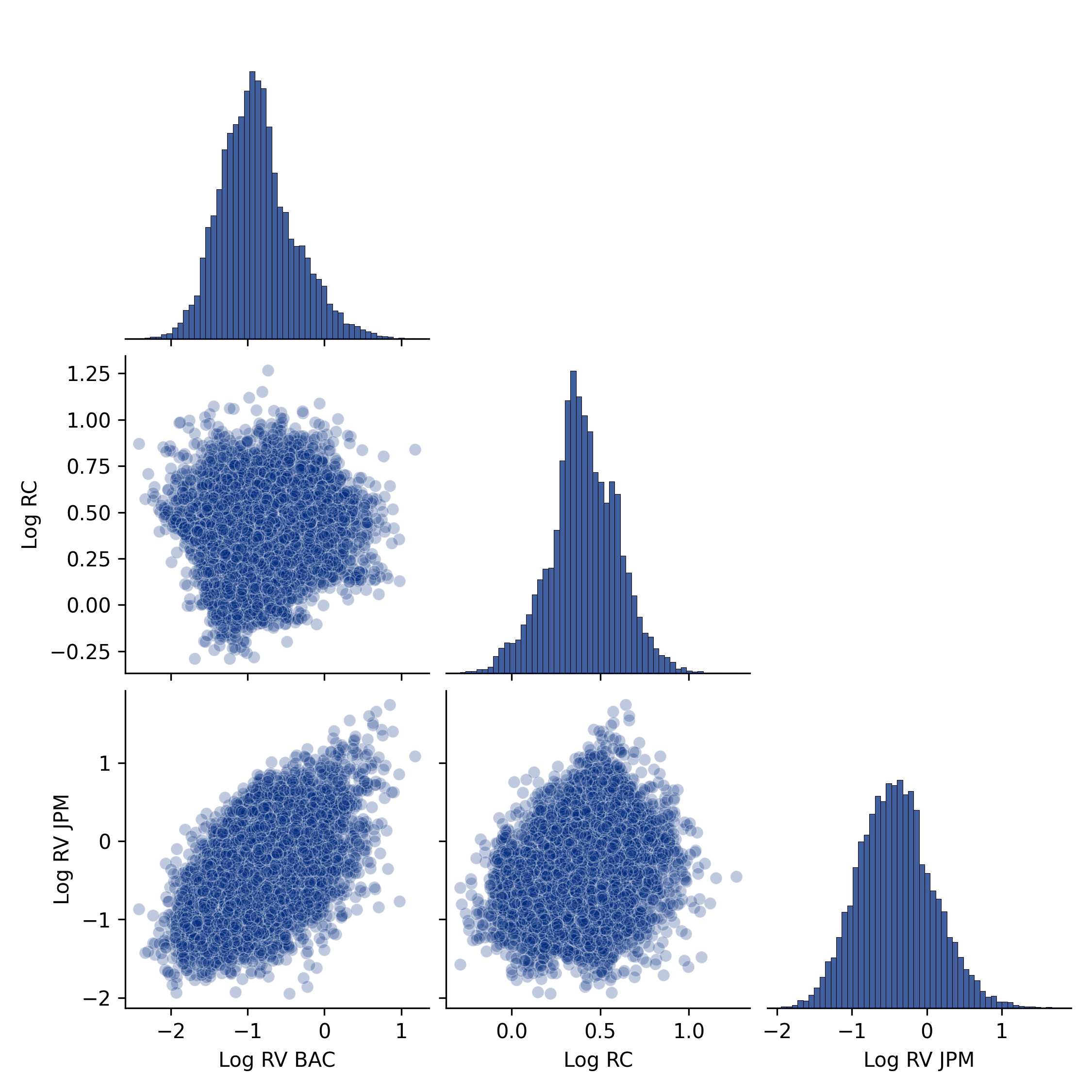}%
\end{minipage}%
\begin{minipage}[t]{0.45\columnwidth}%
\includegraphics[width=1\columnwidth]{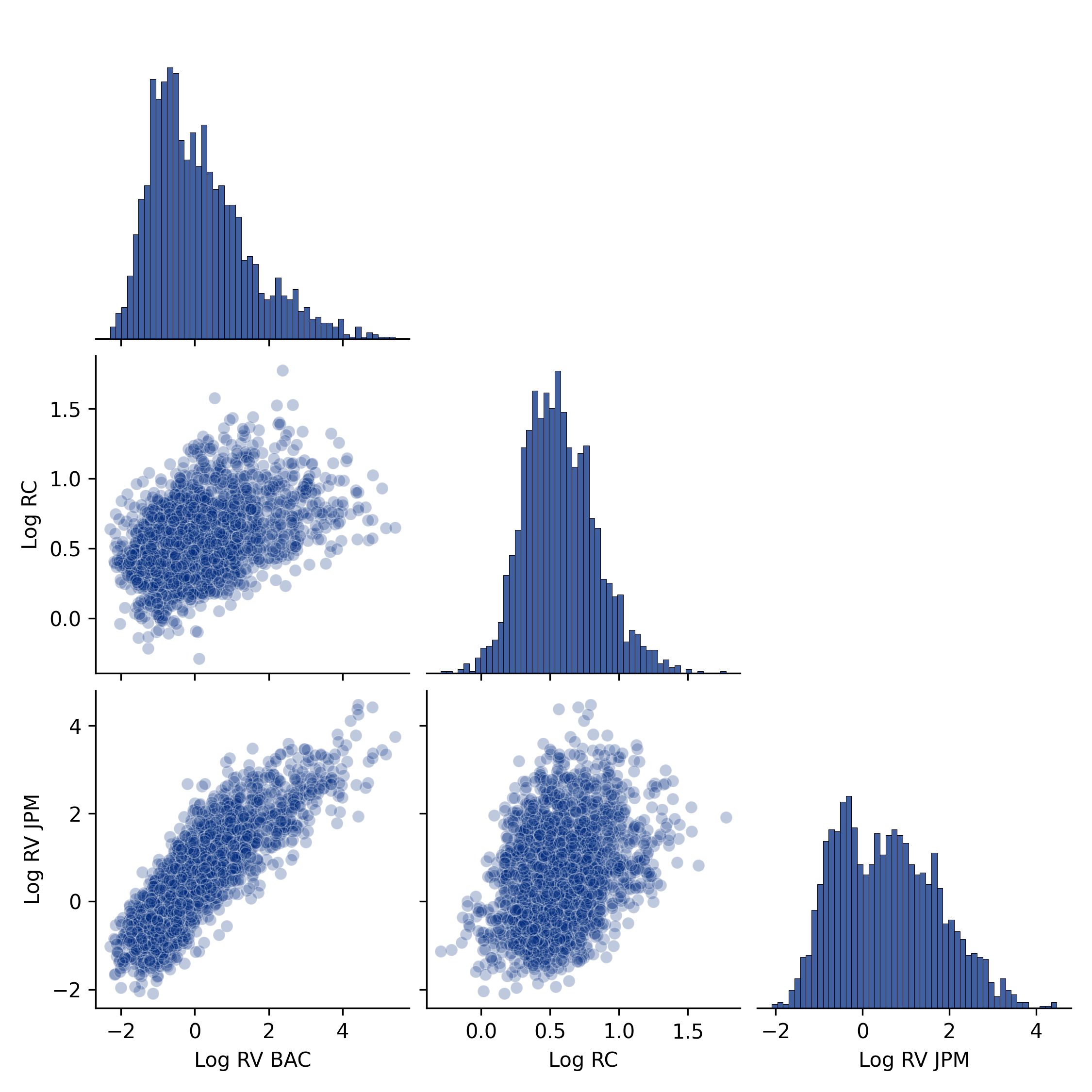}%
\end{minipage}

\caption{The scatterplot matrices display the three unique components of the
$2\times2$ (log) realized covariance matrix: the (log) realized variance
of BAC, the (log) realized variance of JPM, and their (log) realized
covariance. The diagonal panels plot the marginal distributions for
each component. The off-diagonal scatter plots illustrate the pairwise
joint distributions. Top left: Unconditional empirical distribution
of the daily realized covariance matrix for Bank of America (BAC)
and JPMorgan Chase (JPM). Top right: Unconditional empirical distribution
of the matrix logarithm of the realized covariance matrix. Bottom
left: GPD-based predictive distribution of the matrix logarithm of
the realized covariance matrix at the end of the sample. This represents
the model's predictive distribution for the next day ($n+1$), constructed
by simulating draws from the estimated GPD generator conditional on
the final observation in the dataset ($X=Y_{n}$). Bottom right: Unconditional
distribution implied by the estimated GPD model. This implied distribution
is constructed by aggregating simulated draws from the generator,
$G_{\hat{\gamma}_{n}}(Z_{t},X_{t})$, evaluated over the entire empirical
sequence of historical predictors $X_{t}$ using independent standard
normal innovations $Z_{t}$. }

\label{fig:rcov}
\end{figure}

For illustration, the GPD-based predictive distribution (Figure \ref{fig:rcov},
bottom left) is constructed conditional on lagged realized covariances
at the end of the sample. Such distribution is obtained by simulating
from the estimated generator $G_{\hat{\gamma}_{n}}\left(Z,X\right)$,
with $X=Y_{n}$. Marginal PITs computed from the GPD predictive distribution
indicate an adequate fit, suggesting that the model provides a coherent
description of the predictive, marginal distributions for each component
(see Figure \ref{fig:rcov_pit}).

\begin{figure}[tb]
\centering

\begin{tabular}{c}
\includegraphics[width=0.35\textheight]{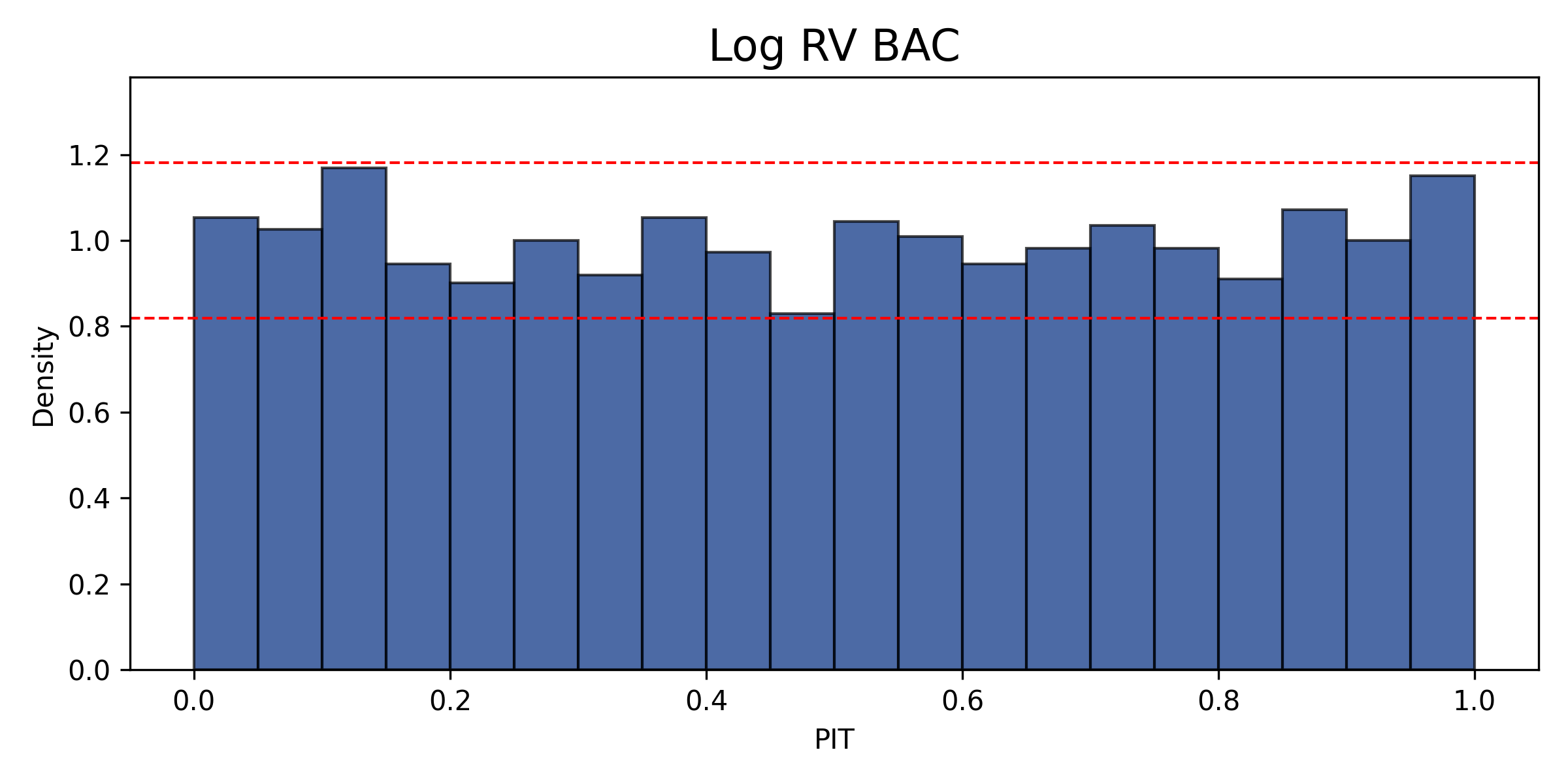}\tabularnewline
\includegraphics[width=0.35\textheight]{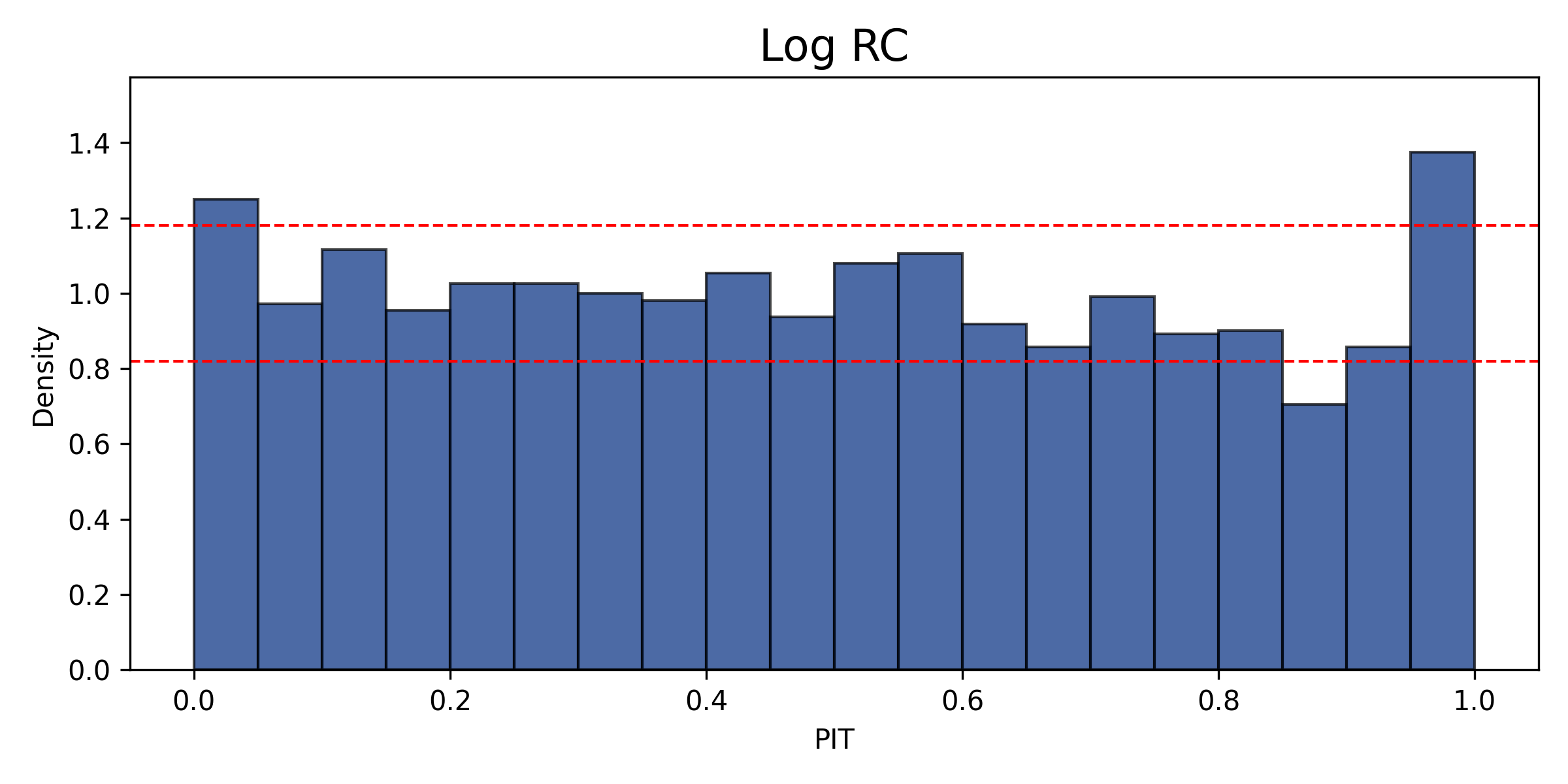}\tabularnewline
\includegraphics[width=0.35\textheight]{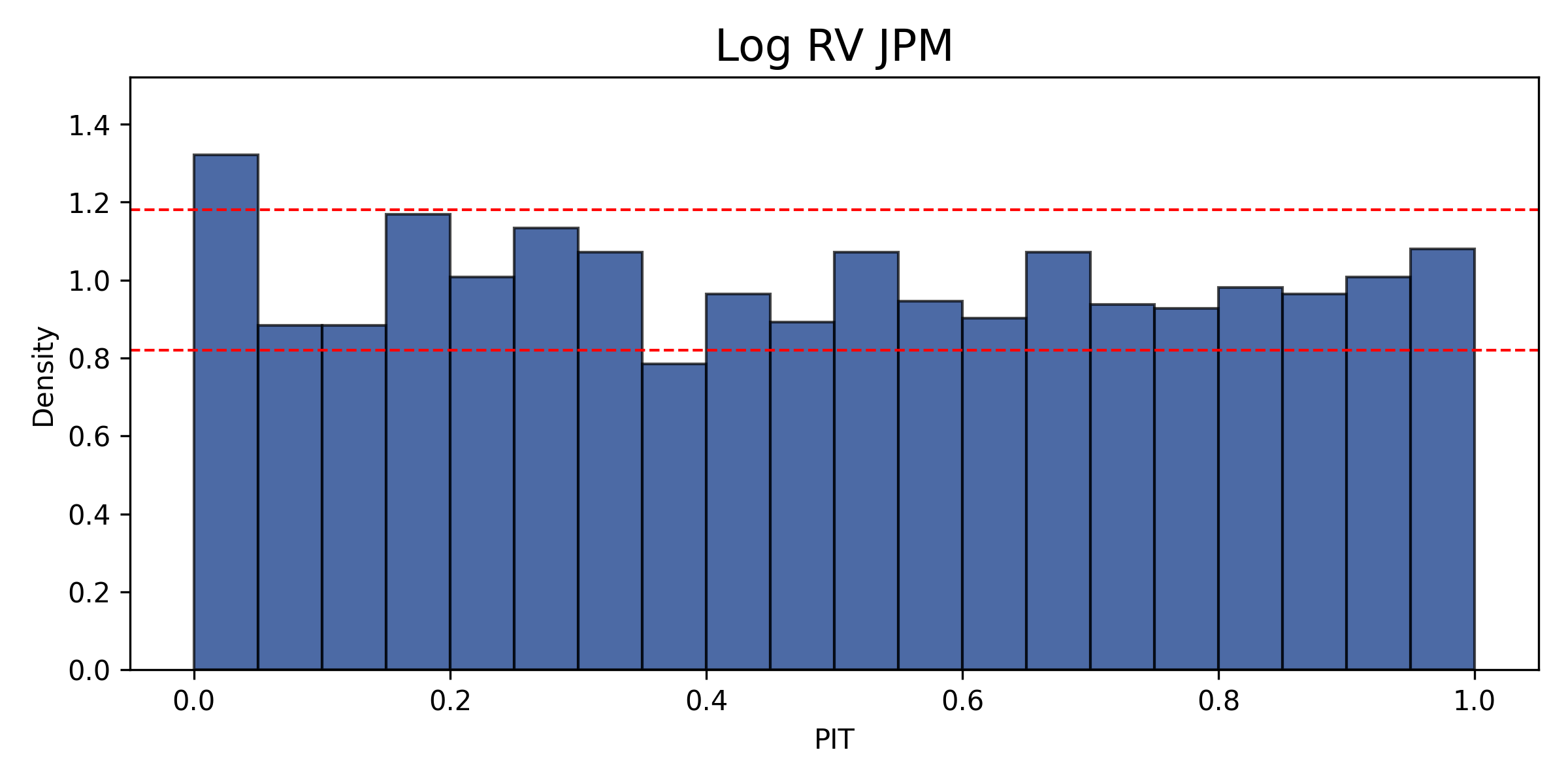}\tabularnewline
\end{tabular}

\caption{Marginal PITs of GPD-based predictive distribution.\protect \\
~\protect \\
~}

\label{fig:rcov_pit}
\end{figure}

Overall, these results highlight the ability of the proposed framework
to produce realistic predictive distributions of realized covariances.
As an additional sanity check, we compute the unconditional distribution
implied by the GPD-based predictive distribution. To do this, we compute
$\left\{ G_{\hat{\gamma}_{n}}\left(Z_{t},X_{t}\right)\right\} _{t=1}^{n}$,
where $Z_{t}$ are again drawn independently from a standard normal
distribution of dimension 3. In words, we use the estimated generator
to draw from the conditional distribution, with conditioning variables
following their unconditional empirical distribution. As it is evident
from Figure \ref{fig:rcov} (bottom right) and Table \ref{tab:rcov_stats},
the unconditional distribution implied by GPD provides an excellent
match to the empirical one.

\section{Concluding remarks\label{sec:conclusion}}

This paper introduced the Generative Predictive Distribution (GPD)
framework as a flexible, model-free, and assumption-lean approach
for modeling predictive distributions of stationary time series. Theoretical
foundation for this approach was given in the form of the generative
representation of a predictive distribution (our Lemma \ref{lem:representation}),
and Hausdorff consistency in the GPD minimax estimation problem was
established (Theorem \ref{thm:consistent}). Formal justification
for the choice of criterion function was also provided (Lemma \ref{lem:criterion}).
Across three distinct empirical applications --- dynamic portfolio
allocation with S\&P 500 returns, forecasting realized variance, and
modeling multivariate realized covariance matrices --- we demonstrated
that GPD can successfully capture non-Gaussianity, conditional heteroskedasticity,
and nonlinear dependencies. A key practical advantage of this framework
is its reliance on direct sampling: because the estimated generator
acts as an efficient simulator, practitioners can easily obtain predictive
distributions for nonlinear transformations of the data. Overall,
the proposed framework provides researchers with a rigorous and computationally
efficient methodology for modeling potentially complex and multivariate
predictive distributions.

Several future research topics could be entertained. We discussed
the challenges of showing weak convergence of the estimated generative
representation in Section 4.3, and a rigorous treatment of this issue
would be of interest. The generator, discriminator, and criterion
function we employed in this paper worked very well in our applications,
but alternative formulations for these could certainly be considered.
We illustrated the usefulness and versatility of the GPD approach
in three practical applications in financial econometrics, but different
applications in econometrics and other fields could be explored further.
In particular, it would be interesting to see how well the GPD method
works in applications where the dimension of $Y_{t}$ is notably larger
than three.

\newpage{}

\appendix

\section{Appendix}

\subsection{Proof of Lemma \ref{lem:representation}\label{subsec:representation}}

First note that the existence of a probability kernel $\mu:\mathcal{X}\to\mathcal{Y}$
satisfying 
\[
\mu(X_{t},B)=P\left\{ Y_{t}\in B\mid X_{t}\right\} \quad\text{a.s.}
\]
for all $B\in\mathcal{B}(\mathcal{Y})$ is ensured by Theorem 8.5
of \citet{Kallenberg:2021}. Moreover, by the same Theorem, the kernel
$\mu$ is unique (almost everywhere with respect to the distribution
of $X_{t}$), and therefore the assumed stationarity of $\left\{ Y_{t},X_{t}\right\} _{t\in\mathbb{Z}}$
implies that the kernel $\mu$ does not depend on $t$. Next, by Lemma
4.22 of \citet{Kallenberg:2021}, there exists a measurable function
$f:\mathcal{X}\times[0,1]\to\mathcal{Y}$ such that for any sequence
of IID $U(0,1)$ random variables $\left\{ U_{t}\right\} _{t\in\mathbb{Z}}$
such that $U_{t}$ is independent of $X_{t}$ it holds that $f(x,U_{t})$
has distribution $\mu(x,\cdot)=P_{Y_{t}\mid X_{t}=x}$ for every $x\in\mathcal{X}$.
As in the proof of Lemma 2.1 of \citet{Zhou:etal:2023}, specifically
choose the sequence $U_{t}$ to be $U_{t}=\Phi(Z_{t,1})$ where $\Phi$
denotes the standard normal cdf and $Z_{t,1}$ the first component
of $Z_{t}$. Defining the function $G:\mathbb{R}^{d_{Z}+d_{X}}\to\mathbb{R}^{d_{Y}}$
as $G(z,x)=f(x,\Phi(z_{1}))$ where $z_{1}$ is the first component
of $z$, it holds that 
\[
G(Z_{t},x)=f(x,U_{t})\sim\mu(x,\cdot)=P_{Y_{t}\mid X_{t}=x},\quad x\in\mathcal{X}.
\]
Finally, defining $V_{t}=G(Z_{t},X_{t})=f(X_{t},U_{t})$ it follows
from the proof of Theorem 8.17 in \citet{Kallenberg:2021} that
\[
(G(Z_{t},X_{t}),X_{t})=(V_{t},X_{t})\overset{d}{=}(Y_{t},X_{t}),
\]
completing the proof.

\subsection{Example illustrating the continuity of $G$\label{subsec:continuity}}

In Section \ref{subsec:generator} we remarked that a continuous generative
map $G$ can be constructed if the joint distribution of $(Y_{t},X_{t})$
admits a strictly positive and continuous density on a connected support.
To illustrate this construction without overly burdensome notation,
we consider the case $d_{X}=1$ and $d_{Y}=2$, though the derivation
naturally generalizes to arbitrary finite dimensions. Assume the target
variable $Y_{t}=(Y_{1,t},Y_{2,t})$ and covariates $X_{t}$ admit
a joint density $f_{Y_{1},Y_{2},X}(y_{1},y_{2},x)$ that is strictly
positive and continuous on a connected support $\mathcal{Y}\times\mathcal{X}\subseteq\mathbb{R}^{3}$.
We first define the forward Rosenblatt transformation $T:\mathcal{Y}\times\mathcal{X}\to(0,1)^{3}$
given by\textcolor{red}{{} }$T(y_{1},y_{2},x)=(u_{1},u_{2},u_{3})$
with
\begin{align*}
u_{1} & =\int_{-\infty}^{x}f_{X}(s)ds,\quad u_{2}=\int_{-\infty}^{y_{1}}f_{Y_{1}\mid X}(s\mid x)ds,\quad u_{3}=\int_{-\infty}^{y_{2}}f_{Y_{2}\mid Y_{1},X}(s\mid y_{1},x)ds.
\end{align*}
Because the joint density is continuous, its partial integrals are
continuously differentiable, meaning that $T$ is itself continuously
differentiable. The Jacobian $J_{T}$ is lower-triangular, yielding
\[
\det(J_{T})=f_{X}(x)\cdot f_{Y_{1}\mid X}(y_{1}\mid x)\cdot f_{Y_{2}\mid Y_{1},X}(y_{2}\mid y_{1},x)=f_{Y_{1},Y_{2},X}(y_{1},y_{2},x)>0
\]
everywhere on the support of $\mathcal{Y}\times\mathcal{X}$. Therefore,
by the multivariate inverse function theorem, the transformation $T$
is a diffeomorphism, which guarantees the existence of a unique, continuously
differentiable inverse map $T^{-1}:(0,1)^{3}\to\mathcal{Y}\times\mathcal{X}$.
We define $G$ as follows:
\[
G\left(z,x\right)=\Pi_{\mathcal{Y}}\left(T^{-1}\left(F_{X}(x),\Phi(z_{1}),\Phi(z_{2})\right)\right),
\]
where $\Pi_{\mathcal{Y}}$ is the projection operator onto the $Y$
coordinates. The map $G$ is continuous.

A particularly illuminating case is given by the Gaussian distribution.
Suppose that $P_{Y_{t}\mid X_{t}=x}=\mathcal{N}\left(\mu(x),\Sigma(x)\right)$,
where
\[
\mu(x)=\left[\begin{array}{c}
\mu_{1}(x)\\
\mu_{2}(x)
\end{array}\right],\quad\Sigma(x)=\left[\begin{array}{cc}
\sigma_{1}^{2}(x) & \sigma_{1}(x)\sigma_{2}(x)\rho(x)\\
\sigma_{1}(x)\sigma_{2}(x)\rho(x) & \sigma_{2}^{2}(x)
\end{array}\right].
\]
Suppose $\rho(x)\in\left(-1,1\right)$, and that $\mu(x)$ and $\Sigma(x)$
are continuous in $x$. Let 
\[
G(Z,x)=\mu(x)+L(x)Z,
\]
where $L(x)$ is the Cholesky factor of $\Sigma(x)$. Clearly, such
a map is continuous.

\subsection{Proof of Lemma \ref{lem:criterion}\label{subsec:criterionproof}}

First, by Fubini's theorem, we write
\[
\mathcal{L}(G,D)-\ln2=\int_{\mathcal{X}}P_{X}(dx)\underset{I(x)}{\underbrace{\int_{\mathcal{Y}}\int_{\mathcal{Y}}\ln\sigma\left(D(y,x)-D(g,x)\right)\mu(x,dy)\mu_{G}(x,dg)}}.
\]
The integral $I(x)$ can be re-written as
\begin{align*}
I(x) & =\frac{1}{2}\int_{\mathcal{Y}}\int_{\mathcal{Y}}\ln\underset{s}{\underbrace{\sigma\left(D(y,x)-D(g,x)\right)}}\mu(x,dy)\mu_{G}(x,dg)\\
 & \quad+\frac{1}{2}\int_{\mathcal{Y}}\int_{\mathcal{Y}}\ln\sigma\left(D(g,x)-D(y,x)\right)\mu(x,dg)\mu_{G}(x,dy)\\
 & =\frac{1}{2}\int_{\mathcal{Y}}\int_{\mathcal{Y}}\ln(s)\mu(x,dy)\mu_{G}(x,dg)+\frac{1}{2}\int_{\mathcal{Y}}\int_{\mathcal{Y}}\ln(1-s)\mu(x,dg)\mu_{G}(x,dy),
\end{align*}
where $s=s(x,y,g)=\sigma(D(y,x)-D(g,x))$, and we have used the identity
$\sigma(-u)=1-\sigma(u)$. We first represent both terms under the
common reference measure $\nu=\frac{1}{2}(\mu+\mu_{G})$. By definition,
both $\mu$ and $\mu_{G}$ are absolutely continuous with respect
to $\nu$. By the Radon-Nikodym theorem, we may re-write $I(x)$ under
the product reference measure $\nu(x,dy)\nu(x,dg)$:
\begin{align*}
I(x)= & \frac{1}{2}\int_{\mathcal{Y}}\int_{\mathcal{Y}}\left[\ln(s)a(y,g)+\ln(1-s)b(y,g)\right]\nu(x,dy)\nu(x,dg),
\end{align*}
where
\[
a(y,g)=\frac{d\mu(x,\cdot)}{d\nu(x,\cdot)}(y)\frac{d\mu_{G}(x,\cdot)}{d\nu(x,\cdot)}(g),\quad\text{and}\quad b(y,g)=\frac{d\mu(x,\cdot)}{d\nu(x,\cdot)}(g)\frac{d\mu_{G}(x,\cdot)}{d\nu(x,\cdot)}(y).
\]
 We maximize the integrand for every pair $(y,g)$. For a fixed pair
$(y,g)$, setting the derivative of the term in square brackets with
respect to $s$ to zero yields
\[
\frac{a(y,g)}{s}-\frac{b(y,g)}{1-s}=0\implies s^{*}=\frac{a(y,g)}{a(y,g)+b(y,g)}.
\]
Since $s=\sigma(D(y,x)-D(g,x))$, we apply the inverse of $\sigma$
(i.e. the logistic transformation) to solve for the optimal functional
difference:
\[
D^{*}(y,x)-D^{*}(g,x)=\ln\frac{s^{*}}{1-s^{*}}=\ln\frac{a(y,g)}{b(y,g)}.
\]
Plugging in the definitions of $a(y,g)$ and $b(y,g)$ and rearranging
terms, we obtain
\begin{align*}
D^{*}(y,x)-D^{*}(g,x) & =\ln\frac{\frac{d\mu(x,\cdot)}{d\nu(x,\cdot)}(y)}{\frac{d\mu_{G}(x,\cdot)}{d\nu(x,\cdot)}(y)}-\ln\frac{\frac{d\mu(x,\cdot)}{d\nu(x,\cdot)}(g)}{\frac{d\mu_{G}(x,\cdot)}{d\nu(x,\cdot)}(g)}.
\end{align*}
It follows that the supremum is attained by any measurable function
of the form
\[
D^{*}(y,x)=\ln\frac{\frac{d\mu(x,\cdot)}{d\nu(x,\cdot)}(y)}{\frac{d\mu_{G}(x,\cdot)}{d\nu(x,\cdot)}(y)}+C(x)=\ln\frac{d\mu(x,\cdot)}{d\nu(x,\cdot)}(y)-\ln\frac{d\mu_{G}(x,\cdot)}{d\nu(x,\cdot)}(y)+C(x),
\]
which proves that \eqref{eq:D*} holds. Plugging the optimal choice
$s=s^{*}$ into $I(x)$, we have
\begin{align*}
I^{*}(x) & =\frac{1}{2}\int_{\mathcal{Y}}\int_{\mathcal{Y}}\ln\frac{a(y,g)}{a(y,g)+b(y,g)}a(y,g)\nu(x,dy)\nu(x,dg)\\
 & \quad+\frac{1}{2}\int_{\mathcal{Y}}\int_{\mathcal{Y}}\ln\frac{b(y,g)}{a(y,g)+b(y,g)}b(y,g)\nu(x,dy)\nu(x,dg)\\
 & =\frac{1}{2}\int_{\mathcal{Y}}\int_{\mathcal{Y}}\ln\frac{a(y,g)}{\frac{a(y,g)+b(y,g)}{2}}a(y,g)\nu(x,dy)\nu(x,dg)\\
 & \quad+\frac{1}{2}\int_{\mathcal{Y}}\int_{\mathcal{Y}}\ln\frac{b(y,g)}{\frac{a(y,g)+b(y,g)}{2}}b(y,g)\nu(x,dy)\nu(x,dg)\\
 & \quad-\ln2.
\end{align*}
Note that
\[
P_{1}(dx,dy,dg)=P_{X}(dx)a(y,g)\nu(x,dy)\nu(x,dg)
\]
 and
\[
P_{2}(dx,dy,dg)=P_{X}(dx)b(y,g)\nu(x,dy)\nu(x,dg),
\]
so integrating $I^{*}(x)$ over $\mathcal{X}$ with respect to $P_{X}$
we obtain
\begin{align*}
 & \frac{1}{2}\int\ln\frac{dP_{1}}{dM}(x,y,g)P_{1}(dx,dy,dg)+\frac{1}{2}\int\ln\frac{dP_{2}}{dM}(x,y,g)P_{2}(dx,dy,dg)-\ln2\\
 & =\frac{1}{2}\mathrm{KL}\left(P_{1}\parallel M\right)+\frac{1}{2}\mathrm{KL}\left(P_{2}\parallel M\right)-\ln2:=\mathrm{JS}\left(P_{1}\parallel P_{2}\right)-\ln2,
\end{align*}
where $M=\frac{1}{2}(P_{1}+P_{2})$. Therefore, equation \eqref{eq:JS}
holds, and the proof is complete.

\subsection{Sliced Wasserstein Distance (SWD)\label{subsec:swd}}

SWD is a fast method for computing distances between two multivariate
distributions, $P$ and $Q$. We have samples $\left\{ V_{i}\right\} _{i=1}^{n}$
and $\left\{ U_{i}\right\} _{i=1}^{n}$ from $P$ and $Q$, taking
values in $\mathbb{R}^{d}$. First, we project the data from $\mathbb{R}^{d}$
to $\mathbb{R}$ using a random projection:
\[
\tilde{V}_{i}=\sum_{j=1}^{d}w_{j}V_{ij},\quad\tilde{U}_{i}=\sum_{j=1}^{d}w_{j}U_{ij},
\]
where $w_{j}\sim\mathcal{N}\left(0,1\right).$ Second, we compute
the 2-Wasserstein distance in 1D:
\[
\mathcal{W}_{2}\left(P,Q\right)=\frac{1}{n}\sum_{i=1}^{n}\left(\tilde{V}_{(i)}-\tilde{U}_{(i)}\right)^{2},
\]
where $\tilde{V}_{(i)}$ is the $i$th order statistic of the projected
$V'$s and $\tilde{U}_{(i)}$ is defined analogously. Since the projection
is random, we repeat this procedure a number of times and average
out the resulting 2-Wasserstein distances.

We employ the sliced Wasserstein distance (SWD) as a stopping criterion
in Algorithm \ref{algo:gpd}. We monitor this metric averaged across
each epoch. If the metric fails to improve for a certain number of
consecutive epochs (also known as patience parameter), training is
terminated and the generator parameters from the best-performing epoch
are used.

\subsection{Proof of Theorem 1\label{subsec:proof-thm}}

We begin by deriving certain uniform convergence (and rate of convergence)
results for the sample average functions $\widehat{f}_{n}(\gamma,\delta)$.
To this end, we verify the assumptions necessary to invoke Theorem
2.1 in \citet{Arcones:Yu:1994}. First, we show that the class of
functions $F(x,y,z,\theta)$ parameterized by $\theta\in\Theta$,
henceforth denoted $\text{\ensuremath{\mathcal{F}}}$, is a VC-subgraph
class of functions in the sense of \citet[Section 2.6]{vanderVaart:Wellner:2023}.
By our Assumption \ref{asm:criterion}, $G_{\gamma}$ and $D_{\delta}$
are both Multi-Layer Perceptrons (MLPs) as in (\ref{eq:generator})
and (\ref{eq:discriminator}). The MLP specifications for $G_{\gamma}$
and $D_{\delta}$, which have ReLU activations by Assumption \ref{asm:criterion},
belong to the class of VC-subgraph functions by Theorem 7 in \citet{Bartlett:etal:2019}.
The term $D_{\delta}(Y_{t},X_{t})-D_{\delta}(G_{\gamma}(Z_{t},X_{t}),X_{t})$
is a special case of feed-forward neural networks with piecewise polynomial
activation, defined in Theorem 7 in \citet{Bartlett:etal:2019}, since
it is possible to enumerate its computational units in a way such
that each computational unit has connections only from units in earlier
layers. It follows that $D_{\delta}(Y_{t},X_{t})-D_{\delta}(G_{\gamma}(Z_{t},X_{t}),X_{t})$
is also VC-subgraph. Furthermore, since $\ln\sigma(\cdot)$ is a fixed
monotone function, it follows that $\ln\sigma(D_{\delta}(Y_{t},X_{t})-D_{\delta}(G_{\gamma}(Z_{t},X_{t}),X_{t}))$
is also VC-subgraph \citep[Lemma 2.6.20]{vanderVaart:Wellner:2023}.
Therefore, the class of functions $\text{\ensuremath{\mathcal{F}}}$
is VC-subgraph.

Second, we show that the class of functions $\text{\ensuremath{\mathcal{F}}}$
has an envelope function with appropriate moments. By assumptions
\ref{asm:criterion} and \ref{asm:compact}, and Lemma \ref{lem:envelope}
below, there exists a finite constant $C$ such that $\sup_{\theta\in\Theta}\left|F\left(x,y,z,\theta\right)\right|\leq C\left(1+\|x\|+\|y\|+\|z\|\right):=\overline{F}(x,y,z)<\infty$
for each $\left(x,y,z\right)$. Moreover, Assumption \ref{asm:moments}
ensures that this envelope function $\overline{F}$ satisfies $\mathbb{E}[\overline{F}(X_{t},Y_{t},Z_{t})^{p}]<\infty$.
Third, we note that by Assumption \ref{asm:data} the $\beta-$mixing
coefficients $\beta_{k}$ satisfy $k^{p/(p-2)}(\ln k)^{2(p-1)/(p-2)}\beta_{k}\to0$
as $k\to\infty$.

By Theorem 2.1 in \citet{Arcones:Yu:1994} we can now conclude that
$n^{1/2}(\widehat{f}_{n}-f)$ converges weakly to a Gaussian process
in $l^{\infty}(\Theta)$ with $f(\theta)$ continuous in $\theta$.
Therefore it also holds that 
\begin{equation}
\sup_{\theta\in\Theta}|\widehat{f}_{n}(\theta)-f(\theta)|\overset{p}{\rightarrow}0\quad\text{and}\quad\sqrt{n}\sup_{\theta\in\Theta}|\widehat{f}_{n}(\theta)-f(\theta)|=O_{p}\left(1\right)\label{ProofThm1ULLNandRATE}
\end{equation}
with the function $f(\theta)$ continuous in $\theta$.

The proof now continues following the arguments in \citet{Meitz:2024}.
First, we note that
\begin{equation}
\sup_{\theta\in\Theta}|\widehat{Q}_{n}(\theta)-Q(\theta)|\leq2\sup_{\theta\in\Theta}|\widehat{f}_{n}(\theta)-f(\theta)|\stackrel{p}{\to}0\label{eq:bound-on-Q}
\end{equation}
by the former uniform convergence result in (\ref{ProofThm1ULLNandRATE}).
As was noted above, $f\left(\theta\right)$ is continuous on $\Theta$
(which is compact by Assumption \ref{asm:compact}), and by Berge's
maximum theorem the function $\sup_{\delta\in\Delta}f(\gamma,\delta)$
is continuous in $\gamma$. This implies that for all $\epsilon>0$
there exists $\eta(\epsilon)>0$ such that
\[
\inf_{\theta\in\Theta\setminus\Theta_{0}^{\epsilon}}Q(\theta)\geq\eta(\epsilon),
\]
where $\Theta_{0}^{\epsilon}=\left\{ \theta\in\Theta:d(\theta,\Theta_{0})\leq\epsilon\right\} $.
For every $\epsilon_{d}>0$, there exists $\eta(\epsilon_{d})$ such
that $\inf_{\theta\in\Theta\setminus\Theta_{0}^{\epsilon_{d}}}Q(\theta)\geq\eta(\epsilon_{d})$.
Also, for every $\epsilon_{p}>0$ there is a $n_{\epsilon_{p}}$ such
that for all $n\geq n_{\epsilon_{p}}$ we have $\tau_{n}\leq\eta(\epsilon_{d})/4$
with probability larger than $1-\epsilon_{p}$. We obtain
\begin{align*}
\sup_{\theta\in\widehat{\Theta}_{n}\left(\tau_{n}\right)}Q\left(\theta\right) & \leq\sup_{\theta\in\widehat{\Theta}_{n}\left(\tau_{n}\right)}|\widehat{Q}_{n}(\theta)-Q(\theta)|+\sup_{\theta\in\widehat{\Theta}_{n}\left(\tau_{n}\right)}\widehat{Q}_{n}\left(\theta\right)\\
 & \leq\eta\left(\epsilon_{d}\right)/2<\inf_{\theta\in\Theta\setminus\Theta_{0}^{\epsilon_{d}}}Q\left(\theta\right),
\end{align*}
implying that $\widehat{\Theta}_{n}\left(\tau_{n}\right)\subseteq\Theta_{0}^{\epsilon_{d}}$
and $\sup_{\theta\in\widehat{\Theta}_{n}\left(\tau_{n}\right)}d\left(\theta,\Theta_{0}\right)\leq\epsilon_{d}$
for all $n\geq n_{\epsilon_{p}}$, with probability larger than $1-\epsilon_{p}$.
Thus, part \ref{thm:1a} holds.

For \ref{thm:1b}, by the latter rate of convergence result in (\ref{ProofThm1ULLNandRATE})
and \eqref{eq:bound-on-Q}, we have 
\[
\sup_{\theta\in\Theta}|\widehat{Q}_{n}(\theta)-Q(\theta)|=O_{p}(n^{-1/2}).
\]
By the first part in the statement of part \ref{thm:1b} and since
$\sup_{\theta\in\Theta_{0}}Q\left(\theta\right)=0$, we have that
for any $\epsilon_{p}>0$, we can find a $n_{\epsilon_{p}}$ such
that for all $n\geq n_{\epsilon_{p}}$,
\[
\sup_{\theta\in\Theta_{0}}\widehat{Q}_{n}\left(\theta\right)\leq\sup_{\theta\in\Theta}|\widehat{Q}_{n}(\theta)-Q(\theta)|+\sup_{\theta\in\Theta_{0}}Q\left(\theta\right)=O_{p}\left(1\right)(n^{-1/2}/\tau_{n})\tau_{n}\leq\tau_{n},
\]
with probability at least $1-\epsilon_{p}$. But the event in the
display is equivalent to $\Theta_{0}\subseteq\widehat{\Theta}_{n}\left(\tau_{n}\right)$,
so $\sup_{\theta\in\Theta_{0}}d\left(\theta,\widehat{\Theta}_{n}\left(\tau_{n}\right)\right)=0$.
This implies part \ref{thm:1b} holds, and the proof is complete.

\subsection{Upper bound on envelope\label{subsec:envelope}}
\begin{lem}
\textcolor{teal}{\label{lem:envelope}}Suppose that Assumptions \ref{asm:compact}
and \ref{asm:criterion} hold. Then,
\[
\sup_{\theta\in\Theta}\left|F\left(x,y,z,\theta\right)\right|\leq C(1+\|x\|+\|y\|+\|z\|).
\]
\end{lem}
\begin{proof}
The linear layers $\gamma_{1},\ldots,\gamma_{L_{G}},\gamma_{O}$ in
$G_{\gamma}(z,x)$ are of the form
\[
\gamma_{\ell}(u)=\Gamma_{\ell,0}+\Gamma_{\ell}u,\quad\ell=1,\dots,L_{G},O,
\]
with the dimensions varying from layer to layer. By the triangle inequality
and matrix inequality $\|Wu\|\leq\|W\|\|u\|$ (for any matrix $W$,
we let $\|W\|$ denote the operator norm), we have $\|\gamma_{\ell}(u)\|\leq\|\Gamma_{\ell,0}\|+\|\Gamma_{\ell}\|\|u\|$.
By compactness of $\Gamma$, there is a $C<\infty$ such that
\begin{align*}
\|\gamma_{1}(z,x)\| & \leq C\left(1+\|\left(z,x\right)\|\right)\leq C\left(1+\|z\|+\|x\|\right),\\
\|\gamma_{\ell}(u)\| & \leq C\left(1+\|u\|\right),\quad\ell=2,\dots,L_{G},O.
\end{align*}
In what follows, the value of the constant $C$ is allowed to change
from line to line. Let $a(u)=\max\{u,0\}$, where the operation is
applied element-wise for a vector argument. We have $\|a(u)\|\leq\|u\|$.
Thus,
\begin{align*}
\|a\left(\gamma_{1}(z,x)\right)\| & \leq\|\gamma_{1}(z,x)\|\leq C\left(1+\|z\|+\|x\|\right).
\end{align*}
Next,
\[
\|\gamma_{2}\left(a\left(\gamma_{1}(z,x)\right)\right)\|\leq C\left(1+\|a\left(\gamma_{1}(z,x)\right)\|\right)\leq C\left(1+C\left(1+\|z\|+\|x\|\right)\right)\leq C\left(1+\|z\|+\|x\|\right).
\]
By induction, this implies that for any number of layers, there is
a constant $C<\infty$ such that
\[
\|G_{\gamma}(z,x)\|\leq C\left(1+\|z\|+\|x\|\right).
\]
By the same argument, we have
\[
\left|D_{\delta}(y,x)\right|\leq C\left(1+\|y\|+\|x\|\right).
\]
The following inequality holds for the log-sigmoid: $\left|\ln\sigma(u)\right|\leq\ln2+|u|$
for every $u$. Therefore,
\begin{align*}
\left|F\left(x,y,z,\theta\right)\right| & =\left|\ln\sigma\bigl(D_{\delta}(y,x)-D_{\delta}(G_{\gamma}(z,x),x)\bigr)\right|\\
 & \leq\ln2+\left|D_{\delta}(y,x)-D_{\delta}(G_{\gamma}(z,x),x)\right|\\
 & \leq\ln2+\left|D_{\delta}(y,x)\right|+\left|D_{\delta}(G_{\gamma}(z,x),x)\right|.
\end{align*}
Here the last term can bounded as follows:
\[
\left|D_{\delta}(G_{\gamma}(z,x),x)\right|\leq C(1+\|G_{\gamma}(z,x)\|+\|x\|)\leq C\left(1+\|x\|+\|z\|\right).
\]
Therefore it follows that
\begin{align*}
\left|F\left(x,y,z,\theta\right)\right| & \leq C(1+\|x\|+\|y\|+\|z\|).
\end{align*}
Since the upper bound is independent of $\theta$, taking supremum
over $\Theta$ on both sides of the last display gives the result.
\end{proof}
\newpage{}

\bibliographystyle{natbib}
\bibliography{references}

\end{document}